\documentclass[a4paper,onecolumn,11pt,unpublished]{quantumarticle}
\pdfoutput=1

\usepackage[utf8]{inputenc}
\usepackage[english]{babel}
\usepackage[T1]{fontenc}
\usepackage{amsmath,amssymb,amsthm,mathtools,bm}
\usepackage{microtype}
\usepackage{booktabs}
\usepackage{graphicx}
\usepackage{float}
\usepackage{xcolor}
\usepackage[numbers,sort&compress]{natbib}
\usepackage{hyperref}
\hypersetup{colorlinks=true,allcolors=quantumviolet}

\newtheorem{theorem}{Theorem}
\newtheorem{lemma}{Lemma}
\newtheorem{proposition}{Proposition}
\newtheorem{corollary}{Corollary}
\newtheorem{conjecture}{Conjecture}
\newcommand{\Tr}{\operatorname{Tr}}
\newcommand{\wh}[1]{\widehat{#1}}
\newcommand{\Z}{\mathbb Z}
\newcommand{\ket}[1]{\lvert #1\rangle}
\newcommand{\proj}[1]{\lvert #1\rangle\!\langle #1\rvert}
\newcommand{\supp}{\operatorname{supp}}
\newcommand{\PhiP}{\Phi_p}
\newcommand{\PhiRh}{\Phi_{R,h}}
\newcommand{\PhiReta}{\Phi_{R,\eta}}
\newcommand{\PhiRHeps}{\Phi_{R,H,\varepsilon}}
\newcommand{\CE}{C_{\mathrm E}}
\newcommand{\Psucc}{P_{\mathrm{succ}}}

\begin{document}

\title{Difference-Set Weyl Channels: Exact Capacity, Optimizer Bifurcation, and Scalable Entanglement Separation}
\author{Se-Wan Ji}
\email{sewanji@nsr.re.kr}
\affiliation{The Affiliated Institute of ETRI, Daejeon 34044, Republic of Korea}
\date{August 19, 2026}
\maketitle
\enlargethispage{2pt}

\begin{abstract}
In odd local dimension $D$, complete Wigner positivity gives a stochastic phase-space action on Wigner-nonnegative states, but does not control signed inputs, entanglement across channel uses, or collective decoding. We use a subsystem-resolved Weyl decomposition to expose all equality conditions of the known tensor-stable output-purity bound. Cyclic difference sets are exactly the uniform shift supports that saturate the universal Parseval lower bound on the worst nontrivial collision mode. For factorized shift--phase noise with a uniform shift support $R$ of cardinality $r$ and phase distribution $h$ having no larger collision radius, this yields $S_{\alpha,\min}(\Phi_{R,h}^{\otimes n})=n\log_2r$ for every $n\geq1$ and $0\leq\alpha\leq2$, the unrestricted capacity $C=\log_2(D/r)$, and a finite-blocklength strong converse. We then introduce a balanced bi-difference-set interpolation $h_\varepsilon=(1-\varepsilon)q_H+\varepsilon u_D$ for equal-size difference sets $R,H$, where $q_H$ and $u_D$ are the uniform distributions on $H$ and $\Z_D$, respectively. Its unassisted capacity is exactly constant on the full interval $0\leq\varepsilon\leq1$, although the Choi state is negative under partial transpose for every $\varepsilon<1$ and the endpoint $\varepsilon=1$ is entanglement breaking. At $\varepsilon=0$ we completely classify the collision degeneracy: every all-use minimum-output state is a product whose local factors lie in one of two mutually unbiased Weyl bases; for $\varepsilon>0$ only the computational basis remains. For Singer parameters $D=q^2+q+1$ and $r=q+1$, the entanglement-assisted advantage is scalable, with $C_{\mathrm E}/C\to2-\varepsilon$ at fixed $\varepsilon$. Finally, for the identity--dephasing profile $h_\eta$ we determine the exact tensor-power collision-entropy phase diagram on the entire noise interval and derive rigorous capacity bounds from computational- and Fourier-basis codes. Above the proved threshold we isolate a distinct von Neumann crossover and state, but do not assume, a tensor-power R\'enyi conjecture supported by diagnostic numerics. Thus complete Wigner positivity coexists with persistent channel entanglement, rigid classical unassisted coding, and a growing entanglement-assisted rate advantage.
\end{abstract}

\section{Introduction}

The classical capacity of a quantum channel is generally a regularized quantity,
\begin{equation}
 C(\Phi)=\lim_{n\to\infty}\frac1n\chi(\Phi^{\otimes n}),
 \label{eq:capacity}
\end{equation}
where $\chi$ is the Holevo information optimized over all input ensembles. Entangled codewords can make the regularization nontrivial~\cite{Holevo1998,Schumacher1997,Hastings2009}. Entanglement-breaking channels are additive~\cite{Shor2002}, but this does not address channels that retain entanglement with a reference. A more discriminating question is whether a non-entanglement-breaking channel can nevertheless have an exactly classical optimal code at every block length.

Odd-dimensional discrete Wigner functions provide a natural test. Random Weyl displacements act as Markov convolutions in phase space, and the resulting channels are completely positive-Wigner-preserving (CPWP)~\cite{Gross2006,Wang2019}. Yet unrestricted inputs may have signed Wigner functions and may be entangled across arbitrarily many channel uses. A positive phase-space kernel therefore does not reduce the quantum coding problem to a classical one: every pure state has the same phase-space $\ell_2$ norm, irrespective of Wigner negativity.

The relevant Weyl-covariant output-$2$-norm estimate and multiplicativity mechanism are known~\cite{FukudaHolevo2006,FukudaGour2017}, as is the covariance route from normalized-projector outputs to classical capacity~\cite{Wolf2005}. Additive minimum-output R\'enyi entropy and covariance also imply a strong converse by general methods~\cite{KonigWehner2009}. Non-entanglement-breaking channels with additive classical capacity are known, most prominently in the depolarizing setting~\cite{King2003}. The unresolved tasks addressed here are instead structural: identify the sparse anisotropic CPWP kernels that saturate the collision certificate, determine every tensor-power equality case including a degenerate endpoint, and construct a scalable family on which unassisted capacity is insensitive to an entanglement-breaking transition while entanglement assistance retains a macroscopic advantage. Unlike recent environment-measurement methods for quantum convolutional channels~\cite{Xiong2026}, we begin with a classical Wigner collision spectrum and flatten its nonconstant modes using cyclic design theory.

The central phenomenon is an optimizer and resource separation. On the balanced family, the unrestricted regularized capacity is constant and attained by a local basis protocol across the full interpolation, including its entanglement-breaking endpoint. Nevertheless, the Choi state remains NPT before that endpoint, the entanglement-assisted rate varies from twice the unassisted rate to the same rate, and the complete optimizer set bifurcates at the two-axis collision degeneracy. The strong converse is retained as a direct finite-block operational consequence rather than treated as the source of novelty.

\subsection{Contributions and assumptions}
\label{sec:contributions}

All entropies and capacities use base-two logarithms. The local dimension $D\geq3$ is odd, and $n$ always denotes a positive integer number of channel uses. For every nonempty $A\subseteq\Z_D$, let $\mathbf1_A$ be its indicator and let $q_A:=\mathbf1_A/|A|$ be the uniform distribution on $A$; write $u_D:=q_{\Z_D}$. Whenever the shift support is denoted by $R$, we set $r:=|R|$; in the balanced construction $|H|=r$ as well. Our main results are as follows.
\begin{enumerate}
 \item We give a subsystem-resolved derivation of the known tensor-stable collision bound. Its reduced-purity form makes all equality conditions explicit even for signed Wigner functions and inputs entangled across channel uses.
 \item For factorized noise with a uniform shift support $R\subset\Z_D$, we prove that cyclic difference sets are exactly the supports saturating the universal lower bound on the worst nontrivial shift collision mode and hence, whenever they exist, are global minimizers. When the phase collision radius is no larger, the bound is saturated and
 \begin{align}
  S_{\alpha,\min}(\Phi_{R,h}^{\otimes n})&=n\log_2r,\notag\\
  \chi(\Phi_{R,h}^{\otimes n})&=n\log_2(D/r),
  \qquad C(\Phi_{R,h})=\log_2(D/r)
  \label{eq:introexact}
 \end{align}
 for every $n$ and $0\leq\alpha\leq2$.
 \item The R\'enyi band yields the direct one-shot converse
 \begin{equation}
  \Psucc^{(n)}(M;\PhiRh)\leq
  \min\!\left\{1,\left[\frac{(D/r)^n}{M}\right]^{(\alpha-1)/\alpha}\right\},
  \qquad 1<\alpha\leq2,
  \label{eq:introstrong}
 \end{equation}
 for every $M$-message code, with arbitrary entangled codewords and collective decoding, where $\Psucc^{(n)}$ denotes the optimal average decoding success probability. Hence the capacity obeys a strong converse, with an explicit exponent at every finite block length.
 \item Under a strict spectral gap, every minimum-output R\'enyi state in the proved band and every Holevo-optimal block ensemble is a computational-basis product construction. At a balanced two-axis degeneracy we close the endpoint problem: every minimizer is a product of computational- or Fourier-basis states, with the basis chosen independently at each site, and all optimal ensembles are classified by their average state.
 \item Let $R$ and $H$ be equal-size cyclic difference sets and define $h_\varepsilon=(1-\varepsilon)q_H+\varepsilon u_D$. The unassisted capacity remains $\log_2(D/r)$ for the entire interval. The channel is entanglement breaking exactly at $\varepsilon=1$ and has an NPT Choi state for every $\varepsilon<1$, while
 \begin{equation}
  \CE(\PhiRHeps)=2\log_2D-\log_2r-H(h_\varepsilon).
  \label{eq:introEA}
 \end{equation}
 \item For every prime power $q$, a cyclic Singer difference set with
 \begin{equation}
  D=q^2+q+1,\qquad r=q+1
  \label{eq:introsinger}
 \end{equation}
 realizes all conclusions on the full interpolation. For fixed $0\leq\varepsilon<1$ the separation is asymptotically macroscopic:
 \begin{equation}
  \lim_{q\to\infty}\frac{\CE(\PhiRHeps)}{C(\PhiRHeps)}=2-\varepsilon.
  \label{eq:introratio}
 \end{equation}
 The dimensions are always odd and need not be prime powers themselves.
 \item For the identity--dephasing profile $h_\eta=\eta\delta_0+(1-\eta)u_D$, we determine $S_{2,\min}(\Phi_{R,\eta}^{\otimes n})$ and every minimizer on the full interval: computational-basis products below the collision threshold, Fourier-basis products above it, and arbitrary sitewise choices between the two bases at the threshold.
 \item On the same path, we prove an all-use capacity sandwich from the two basis codes and the collision converse.  We identify the unique von Neumann crossover $\eta_{\rm vN}>\eta_\star$ and formulate the exact tensor-power R\'enyi expression above $\eta_\star$ as a conjecture, explicitly separated from every proved capacity statement.
\end{enumerate}

\paragraph{Boundary of the novelty claim.}
The collision bound, its abstract multiplicativity mechanism, the normalized-projector implication for R\'enyi entropies, the covariance capacity formula, and the general R\'enyi route to a strong converse are prior ingredients~\cite{FukudaHolevo2006,FukudaGour2017,Wolf2005,KonigWehner2009}. The additional content is their design-theoretic minimax saturation, the subsystem-resolved all-use equality structure, the complete two-basis endpoint classification, a Fourier-resolved NPT criterion, and the balanced family with a constant unassisted capacity but scalable entanglement-assisted separation. The identities above establish tensor-power additivity for the fixed channels considered here; they do not assert strong additivity with an arbitrary second channel $\Psi$.

The assumptions and logical dependencies are summarized in Table~\ref{tab:scope}. The restriction $\alpha\leq2$ is essential to the present collision method; no minimum-output tensor-power identity or output majorization is claimed above order two. The converse characterizes saturation within uniform-shift product models, not every Weyl channel whose classical capacity single-letterizes.

\begin{table}[t]
\centering
\renewcommand{\arraystretch}{1.16}
\begin{tabular}{@{}p{0.30\textwidth}p{0.64\textwidth}@{}}
\toprule
\textbf{Assumptions} & \textbf{Conclusion} \\
\midrule
Arbitrary odd-dimensional random-Weyl noise & Subsystem-resolved form and equality conditions of the known tensor-stable output-purity bound. \\
Uniform $r$-point shift support & Difference sets are exactly the supports saturating the universal Parseval lower bound on the worst nonconstant shift collision mode; whenever they exist, they are global minimizers. \\
Difference-set support and $\beta_h\leq\alpha_R$ & Exact R\'enyi minima for $0\leq\alpha\leq2$, the tensor-power identity $\chi(\Phi_{R,h}^{\otimes n})=n\log(D/r)$, regularized capacity, and strong converse. \\
Strict gap $\beta_h<\alpha_R$ & Complete classification of entropy minimizers and Holevo-optimal ensembles at every tensor power. \\
Equal-size shift and phase difference sets & Complete endpoint classification into sitewise computational/Fourier products and all corresponding optimal ensembles. \\
$h_\varepsilon=(1-\varepsilon)q_H+\varepsilon u_D$ & Constant exact unassisted capacity on $[0,1]$, NPT for $\varepsilon<1$, entanglement breaking at $\varepsilon=1$, and exact entanglement-assisted capacity. \\
$h_\eta=\eta\delta_0+(1-\eta)u_D$ & Exact all-use collision-entropy phase diagram and optimizer bifurcation on the full interval. \\
$h_\eta$ above the collision threshold & Rigorous two-basis lower bound and collision upper bound on capacity; the proposed all-use von Neumann/R\'enyi formula is stated only as Conjecture~\ref{conj:eta-renyi}. \\
Singer parameters $D=q^2+q+1$, $r=q+1$ & Infinite family with $\CE/C\to2-\varepsilon$ at fixed $\varepsilon$; $D$ is always odd and may be composite. \\
\bottomrule
\end{tabular}
\caption{Logical scope of the results. Here $r=|R|$, $\alpha_R$ is the universal Parseval lower bound on the shift collision radius, and $\beta_h$ is the nontrivial phase collision radius.}
\label{tab:scope}
\end{table}

Two dimension statements should not be conflated. The Wigner-collision framework and its tensor-stable bound are formulated for an arbitrary odd integer $D$, including composite $D$. The exact-capacity theorem is conditional on the existence of a suitable cyclic difference set in $\Z_D$. The Singer construction supplies an explicit infinite subfamily indexed by prime powers $q$, but the physical dimension $D=q^2+q+1$ is not required to be a prime power.

\section{Weyl, characteristic-function, and Wigner conventions}

We keep the local Hilbert space $\mathcal H_D\cong\mathbb C^D$ of odd dimension $D\geq3$ throughout, with computational basis $\{\ket j:j\in\Z_D\}$. Thus all statements about Wigner positivity, including complete Wigner positivity, refer to the cyclic odd-dimensional discrete Wigner representation; no even-dimensional phase-space extension is claimed. Let $\Z_D$ denote integers modulo $D$, set $\omega=e^{2\pi i/D}$, and define
\begin{equation}
 X\ket j=\ket{j+1},\qquad Z\ket j=\omega^j\ket j,
 \qquad W_{a,b}=X^aZ^b,
 \label{eq:weyl}
\end{equation}
for $a,b,j\in\Z_D$. Write $\mathcal P_D:=\Z_D^2$ for the discrete phase space and $W_\xi:=W_{a,b}$ for $\xi=(a,b)\in\mathcal P_D$. Projective phases are irrelevant for the random-unitary channels considered below. Throughout, $I$ denotes the identity operator on the space indicated by context or by a subsystem label. The Weyl operators obey
\begin{equation}
 \Tr(W_{a,b}^{\dagger}W_{u,v})=D\,\delta_{a,u}\delta_{b,v},
 \label{eq:orthogonal}
\end{equation}
where $\delta_{a,u}$ is the Kronecker delta. For an operator $M$, define
\begin{equation}
 \chi_M(u,v)=\Tr(MW_{u,v}^{\dagger}).
 \label{eq:chi}
\end{equation}
We also use the Fourier basis
\begin{equation}
 \ket{\widetilde j}:=\frac1{\sqrt D}\sum_{k\in\Z_D}\omega^{-jk}\ket k,
 \qquad X\ket{\widetilde j}=\omega^j\ket{\widetilde j},
 \label{eq:fourierbasis}
\end{equation}
which is mutually unbiased with the computational basis.
Then
\begin{equation}
 M=\frac1D\sum_{u,v}\chi_M(u,v)W_{u,v},
 \qquad
 \Tr(M^{\dagger}M)=\frac1D\sum_{u,v}|\chi_M(u,v)|^2.
 \label{eq:weylparseval}
\end{equation}
For $n$ systems, write $[n]:=\{1,\ldots,n\}$ and, for any $A\subseteq[n]$, let $A^c:=[n]\setminus A$. Use $\bm z=(z_1,\ldots,z_n)\in\mathcal P_D^n$, define $W_{\bm z}:=\bigotimes_{i=1}^nW_{z_i}$ and $\chi_M(\bm z):=\Tr(MW_{\bm z}^{\dagger})$. Then
\begin{equation}
 \Tr(M^{\dagger}M)=D^{-n}\sum_{\bm z\in\mathcal P_D^n}|\chi_M(\bm z)|^2.
 \label{eq:nparseval}
\end{equation}

Let $p$ be a probability distribution on $\mathcal P_D$ and define
\begin{equation}
 \PhiP(\rho)=\sum_{a,b}p(a,b)W_{a,b}\rho W_{a,b}^{\dagger}.
 \label{eq:channel}
\end{equation}
Conjugation gives
\begin{equation}
 W_{a,b}W_{u,v}W_{a,b}^{\dagger}=\omega^{bu-av}W_{u,v}.
 \label{eq:conjugation}
\end{equation}
With the symplectic Fourier transform
\begin{equation}
 \wh p(u,v)=\sum_{a,b}p(a,b)\omega^{bu-av}.
 \label{eq:sympFT}
\end{equation}
For tensor-product expressions below, when $z_i=(u_i,v_i)$ we abbreviate $\wh p(u_i,v_i)$ as $\wh p(z_i)$. We then obtain
\begin{equation}
 \PhiP(W_{u,v})=\wh p(u,v)W_{u,v}.
 \label{eq:multiplier}
\end{equation}

For completeness, let $\bar 2:=2^{-1}\in\Z_D$ and set $T_{a,b}:=\omega^{\bar 2ab}W_{a,b}$ and $T_x:=T_{a,b}$ for $x=(a,b)$. The phase-point operators and Wigner function are
\begin{equation}
 A_0:=\frac1D\sum_{a,b\in\Z_D}T_{a,b},\qquad
 A_x:=T_xA_0T_x^\dagger,\qquad
 W_\rho(x):=\frac1D\Tr(A_x\rho).
 \label{eq:wignerdefinition}
\end{equation}
For multipartite systems we use the tensor-product phase-point basis. A state is Wigner nonnegative when all entries of this representation are nonnegative. Following Ref.~\cite{Wang2019}, a channel is called CPWP when, for every odd-dimensional reference system and every joint Wigner-nonnegative state, its extension by the identity produces another Wigner-nonnegative state. We use this definition with the cyclic $\Z_D$ representation of Ref.~\cite{Gross2006} for every odd $D$, including composite $D$.

In this representation, Weyl conjugation is translation. Hence
\begin{equation}
 W_{\PhiP(\rho)}(y)=\sum_{x\in\mathcal P_D}K_p(y|x)W_\rho(x),
 \qquad K_p(y|x)=p(y-x)\geq0,\quad y\in\mathcal P_D.
 \label{eq:wignerkernel}
\end{equation}
The kernel is doubly stochastic. Adjoining a reference replaces it by $K_p\otimes\delta_{\mathrm{ref}}$, where $\delta_{\mathrm{ref}}$ is the identity transition kernel on the reference phase space; hence the channel is CPWP directly from the preceding definition. The Fourier eigenvalues of the classical convolution $K_p$ are $\wh p(u,v)$, and those of its collision operator $K_p^{\dagger}K_p$ are $|\wh p(u,v)|^2$. Define
\begin{equation}
 \kappa(p)=\max_{(u,v)\neq(0,0)}|\wh p(u,v)|^2.
 \label{eq:kappa}
\end{equation}
Because $p$ is a probability distribution, $0\leq\kappa(p)\leq1$.

\section{What pointwise cancellation does and does not prove}

Write a signed Wigner function in terms of its positive and negative parts, $W_\rho=W_+-W_-$ with $W_\pm\geq0$ and disjoint supports. For functions $f,g$ on $\mathcal P_D$, use $\|f\|_2^2:=\sum_x|f(x)|^2$ and $\langle f,g\rangle:=\sum_x f(x)^*g(x)$. Since $K_p\geq0$,
\begin{equation}
 |K_pW_\rho|\leq K_p|W_\rho|
 \label{eq:pointwise}
\end{equation}
pointwise. Equivalently,
\begin{align}
 \|K_pW_\rho\|_2^2
 &=\|K_pW_+\|_2^2+\|K_pW_-\|_2^2
 -2\langle K_pW_+,K_pW_-\rangle\notag\\
 &\leq\|K_pW_+\|_2^2+\|K_pW_-\|_2^2.
 \label{eq:cancellation}
\end{align}
Thus negativity can create destructive interference after a positive kernel, although the cross term can also vanish.

This observation is not by itself an optimization theorem. Define the Wigner negativity by $\mathcal N(W_\rho):=\sum_xW_-(x)=[\sum_x|W_\rho(x)|-1]/2$. The absolute value $|W_\rho|$ then has total mass
\begin{equation}
 \sum_x|W_\rho(x)|=1+2\mathcal N(W_\rho),
 \label{eq:manafactor}
\end{equation}
which exceeds one for a negative Wigner function, and its normalization need not represent any quantum state. Moreover, all pure states satisfy
\begin{equation}
 D\sum_xW_\rho(x)^2=\Tr(\rho^2)=1,
 \label{eq:pureL2}
\end{equation}
so Wigner-positive and Wigner-negative pure states have the same input $\ell_2$ norm. The identity channel, whose Wigner kernel is positive, already shows that negativity need not reduce output purity. A tensor-stable no-advantage statement therefore requires the quantum constraints linking Fourier weights across subsystems. The next section supplies exactly those constraints.

\section{Tensor-stable Wigner-collision theorem}

For $\bm z=(z_1,\ldots,z_n)\in\mathcal P_D^n$, define its exact support
\begin{equation}
 \supp(\bm z)=\{i\in[n]:z_i\neq(0,0)\}.
 \label{eq:support}
\end{equation}
For a density operator $\rho$ on $n$ systems, let
\begin{equation}
 E_T(\rho)=\sum_{\supp(\bm z)=T}|\chi_\rho(\bm z)|^2,
 \qquad T\subseteq[n].
 \label{eq:ET}
\end{equation}
All $E_T$ are nonnegative even when $W_\rho$ is negative. This exact-support decomposition is a subsystem-resolved Weyl analogue of the operator-weight bookkeeping used in quantum weight enumerators~\cite{ShorLaflamme1997,Rains1998}. The identity below is standard Weyl Parseval on reduced subsystems; the new use made here is to couple it to the collision spectrum of a positive Wigner kernel.

\begin{lemma}[Cumulative support identity]
\label{lem:cumulative-support}
For every $S\subseteq[n]$,
\begin{equation}
 C_S(\rho):=\sum_{T\subseteq S}E_T(\rho)
 =D^{|S|}\Tr(\rho_S^2),
 \label{eq:cumulative}
\end{equation}
where $\rho_S:=\Tr_{S^c}\rho$.
\end{lemma}

\begin{proof}
The sum on the left contains exactly the Weyl operators that are identity outside $S$; $\bm z_S$ ranges over $\mathcal P_D^{|S|}$ and $I_{S^c}$ denotes the identity on the complementary subsystems:
\begin{equation}
 C_S=\sum_{\bm z_S\in\mathcal P_D^{|S|}}
 \left|\Tr\!\left[\rho(W_{\bm z_S}^{\dagger}\otimes I_{S^c})\right]\right|^2.
\end{equation}
The trace equals $\Tr(\rho_SW_{\bm z_S}^{\dagger})$. Applying Weyl Parseval on the subsystem $S$ gives \eqref{eq:cumulative}.
\end{proof}

\begin{lemma}[Support-binomial identity]
\label{lem:support-binomial}
For $0\leq\kappa\leq1$,
\begin{equation}
 \sum_{T\subseteq[n]}\kappa^{|T|}E_T
 =\sum_{S\subseteq[n]}\kappa^{|S|}(1-\kappa)^{n-|S|}C_S.
 \label{eq:binomialidentity}
\end{equation}
\end{lemma}

\begin{proof}
Insert $C_S=\sum_{T\subseteq S}E_T$ into the right-hand side and interchange sums. The coefficient of a fixed $E_T$ becomes, after writing $U:=S\setminus T\subseteq T^c$,
\begin{align}
 \sum_{S\supseteq T}\kappa^{|S|}(1-\kappa)^{n-|S|}
 &=\kappa^{|T|}\sum_{U\subseteq T^c}\kappa^{|U|}(1-\kappa)^{|T^c|-|U|}\notag\\
 &=\kappa^{|T|}.
\end{align}
\end{proof}

\begin{theorem}[Tensor-stable Wigner-collision bound]
\label{thm:collision-bound}
Let $\PhiP$ be the random-Weyl channel \eqref{eq:channel}. For every $n\geq1$ and every input state $\rho$,
\begin{equation}
 \Tr\!\left[\PhiP^{\otimes n}(\rho)^2\right]
 \leq\gamma_p^n,
 \qquad
 \gamma_p:=\frac1D+\left(1-\frac1D\right)\kappa(p).
 \label{eq:mainbound}
\end{equation}
\end{theorem}

\begin{proof}
Equations \eqref{eq:multiplier} and \eqref{eq:nparseval} give
\begin{equation}
 \Tr[\PhiP^{\otimes n}(\rho)^2]
 =D^{-n}\sum_{\bm z}\left(\prod_{i=1}^n|\wh p(z_i)|^2\right)|\chi_\rho(\bm z)|^2.
 \label{eq:exactpurity}
\end{equation}
Set $\kappa:=\kappa(p)$. The identity multiplier is one and every nonidentity squared multiplier is at most $\kappa$. Therefore
\begin{equation}
 \Tr[\PhiP^{\otimes n}(\rho)^2]
 \leq D^{-n}\sum_T\kappa^{|T|}E_T.
 \label{eq:purityupper}
\end{equation}
By Lemmas~\ref{lem:cumulative-support} and~\ref{lem:support-binomial},
\begin{align}
 \sum_T\kappa^{|T|}E_T
 &=\sum_S\kappa^{|S|}(1-\kappa)^{n-|S|}D^{|S|}\Tr(\rho_S^2)\notag\\
 &\leq\sum_S(D\kappa)^{|S|}(1-\kappa)^{n-|S|}\notag\\
 &=[1+(D-1)\kappa]^n.
 \label{eq:puritybinomial}
\end{align}
Dividing by $D^n$ proves \eqref{eq:mainbound}.
\end{proof}

The proof treats all inputs uniformly. Wigner negativity changes the distribution of the nonnegative quantities $E_T$, but cannot violate the cumulative constraints $C_S=D^{|S|}\Tr(\rho_S^2)\leq D^{|S|}$. Entanglement across uses can only decrease some marginal purities below one. Product pure states can remove this marginal-purity loss simultaneously, explaining the tensor-power form $\gamma_p^n$; saturation additionally requires their nonzero characteristic modes to lie in the top collision eigenspaces.

For a density operator $\sigma$, define $S_\alpha(\sigma):=(1-\alpha)^{-1}\log\Tr(\sigma^\alpha)$ for $\alpha\neq1$, with $S_0(\sigma):=\log\operatorname{rank}\sigma$ and $S_1(\sigma):=-\Tr(\sigma\log\sigma)$, and define $S_{\alpha,\min}(\Phi):=\min_\rho S_\alpha(\Phi(\rho))$.

\begin{corollary}[R\'enyi-entropy band]
\label{cor:renyi-band}
For every $n\geq1$ and $0\leq\alpha\leq2$,
\begin{equation}
 S_{\alpha,\min}(\PhiP^{\otimes n})
 \geq-n\log\gamma_p.
 \label{eq:renyilower}
\end{equation}
\end{corollary}

\begin{proof}
The theorem gives $S_{2,\min}\geq-n\log\gamma_p$. R\'enyi entropy is nonincreasing in its order, so $S_\alpha\geq S_2$ for $0\leq\alpha\leq2$.
\end{proof}

The numerical value of \eqref{eq:mainbound} is the Weyl-covariant output-$2$-norm estimate of Fukuda and Holevo~\cite{FukudaHolevo2006} and also the Weyl-diagonal specialization of the general tensor-stable unital-channel estimate of Fukuda and Gour~\cite{FukudaGour2017}; it is not claimed as a new bound. Here it is derived without introducing a Bloch-space matrix: $\kappa(p)$ is read directly from the classical Wigner collision operator, and all tensor-power equality conditions are expressed through reduced-state purities. Because the argument starts from output purity, Corollary~\ref{cor:renyi-band} covers $0\leq\alpha\leq2$ only. No tensor-power additivity claim for $\alpha>2$, strong-additivity claim with a different channel, or total majorization ordering of output spectra is made.

\section{Uniform shifts and Fourier extremality}

For a function $f:\Z_D\to\mathbb C$, define the one-dimensional Fourier transform $\wh f(k):=\sum_{x\in\Z_D}f(x)\omega^{kx}$. Let $R\subset\Z_D$ be nonempty and set $r:=|R|$, so $1\leq r\leq D$ and $q_R=\mathbf1_R/r$ under the convention above. Classical Parseval gives
\begin{equation}
 \sum_{k\in\Z_D}|\wh q_R(k)|^2
 =D\sum_xq_R(x)^2=\frac Dr.
 \label{eq:classicalparseval}
\end{equation}
Since $\wh q_R(0)=1$,
\begin{equation}
 \frac1{D-1}\sum_{k\neq0}|\wh q_R(k)|^2
 =\frac{D-r}{r(D-1)}=: \alpha_R.
 \label{eq:alpha}
\end{equation}
Consequently
\begin{equation}
 m_R:=\max_{k\neq0}|\wh q_R(k)|^2\geq\alpha_R.
 \label{eq:minimax}
\end{equation}

A subset $R\subset\Z_D$ is a cyclic $(D,r,\lambda)$ difference set when every nonzero $g\in\Z_D$ has exactly $\lambda$ ordered representations $g=a-a'$ with $a,a'\in R$; necessarily $\lambda(D-1)=r(r-1)$.

\begin{lemma}[Difference-set equality condition]
\label{lem:difference-set}
For $R\subset\Z_D$, $|R|=r$, the following are equivalent:
\begin{enumerate}
\item $m_R=\alpha_R$;
\item $|\wh q_R(k)|^2=\alpha_R$ for all $k\neq0$;
\item $R$ is a cyclic $(D,r,\lambda)$ difference set, with $\lambda(D-1)=r(r-1)$.
\end{enumerate}
\end{lemma}

\begin{proof}
The fixed average \eqref{eq:alpha} proves the equivalence of the first two statements. Define the autocorrelation
\begin{equation}
 c_R(g)=\sum_xq_R(x)q_R(x-g)
 =\frac1{r^2}\#\{(a,a')\in R^2:a-a'=g\}.
 \label{eq:autocorrelation}
\end{equation}
Its Fourier transform is $|\wh q_R(k)|^2$. If the latter equals one at zero and $\alpha_R$ elsewhere, Fourier inversion gives
\begin{equation}
 c_R(0)=\frac1r,
 \qquad
 c_R(g\neq0)=\frac{1-\alpha_R}{D}
 =\frac{r-1}{r(D-1)}.
 \label{eq:inverseauto}
\end{equation}
Thus every nonzero $g$ has $r^2c_R(g)=r(r-1)/(D-1)=\lambda$ ordered representations as $a-a'$. The converse follows by Fourier transforming the difference-set autocorrelation.
\end{proof}

Hence difference sets exactly saturate the universal lower bound on the worst nontrivial collision eigenvalue. Whenever the parameters $(D,r)$ admit a cyclic difference set, these and only these supports attain the global minimum $\alpha_R$. If no such difference set exists, the minimum is strictly larger than $\alpha_R$, and Lemma~\ref{lem:difference-set} does not classify the minimizers. This extremal statement is independent of quantum mechanics; quantum structure enters when the collision spectrum is tested on unrestricted codewords through Theorem~\ref{thm:collision-bound}.

\section{Exact capacity for shift--phase noise}

For a channel $\Phi$, let $\chi(\Phi)$ denote the optimized one-use Holevo information, define the regularized classical capacity by $C(\Phi):=\lim_{n\to\infty}n^{-1}\chi(\Phi^{\otimes n})$, and write $S_{\min}:=S_{1,\min}$. The limit exists because $\chi(\Phi^{\otimes n})$ is superadditive and bounded above by $n\log D$, so Fekete's lemma applies. Let $h$ be a probability distribution on $\Z_D$ and define
\begin{equation}
 p(a,b)=q_R(a)h(b),
 \qquad
 \PhiRh(\rho):=\sum_{a,b}q_R(a)h(b)W_{a,b}\rho W_{a,b}^{\dagger}.
 \label{eq:productchannel}
\end{equation}
With this convention,
\begin{equation}
 \wh p(u,v)=\wh h(u)\wh q_R(-v).
 \label{eq:factorization}
\end{equation}
Set
\begin{equation}
 \beta_h=\max_{u\neq0}|\wh h(u)|^2.
 \label{eq:betah}
\end{equation}
The cases $u=0,v\neq0$ and $u\neq0,v=0$ show that
\begin{equation}
 \kappa(p)=\max\{m_R,\beta_h\}.
 \label{eq:exactkappa}
\end{equation}
Indeed, when both indices are nonzero the product is at most $m_R\beta_h\leq\max\{m_R,\beta_h\}$.

For every computational-basis state, the phase operation is immaterial and the shifts are orthogonal:
\begin{equation}
 \PhiRh(\proj j)=\frac1r\sum_{a\in R}\proj{j+a}.
 \label{eq:flatoutput}
\end{equation}
This is a normalized rank-$r$ projector.

\begin{theorem}[Exact tensor-power entropies and classical capacity]
\label{thm:exact-capacity}
Suppose $R$ is a cyclic difference set and $\beta_h\leq\alpha_R$. Then for every $n\geq1$ and $0\leq\alpha\leq2$,
\begin{align}
 S_{\alpha,\min}(\PhiRh^{\otimes n})&=n\log r,
 \label{eq:exactentropy}\\
 \chi(\PhiRh^{\otimes n})&=n\log(D/r),
 \label{eq:exactholevo}\\
 C(\PhiRh)&=\log(D/r).
 \label{eq:exactcapacity}
\end{align}
The rate is attained by product computational-basis states and product computational-basis measurements. Conversely, within uniform $r$-shift product models, equality of the collision certificate $\gamma_p=1/r$ requires $R$ to be a difference set and $\beta_h\leq\alpha_R$.
\end{theorem}

\begin{proof}
Difference-set equality gives $m_R=\alpha_R$, and therefore $\kappa(p)=\alpha_R$. Direct calculation yields
\begin{equation}
 \gamma_p=\frac1D+\left(1-\frac1D\right)\alpha_R=\frac1r.
 \label{eq:gammar}
\end{equation}
The collision theorem gives $S_{2,\min}(\PhiRh^{\otimes n})\geq n\log r$. The product of the inputs in \eqref{eq:flatoutput} produces a normalized rank-$r^n$ projector, whose R\'enyi entropy equals $n\log r$ at every order. This proves \eqref{eq:exactentropy} for $0\leq\alpha\leq2$.

Random-Weyl channels are covariant under the irreducible Weyl representation:
\begin{equation}
 \PhiRh(W_{c,d}\rho W_{c,d}^{\dagger})
 =W_{c,d}\PhiRh(\rho)W_{c,d}^{\dagger}.
 \label{eq:covariance}
\end{equation}
For any ensemble, the output dimension and the definition of minimum output entropy give
\begin{equation}
 \chi(\PhiRh)\leq\log D-S_{\min}(\PhiRh).
 \label{eq:covariantupper}
\end{equation}
If $\rho_\star$ is a minimum-output state, the Weyl twirl satisfies
\begin{equation}
 \frac1{D^2}\sum_{c,d\in\Z_D}W_{c,d}\rho_\star W_{c,d}^{\dagger}
 =\frac ID.
 \label{eq:weyltwirl}
\end{equation}
The channel is unital, so the uniformly weighted Weyl orbit of $\rho_\star$ has average output $I/D$; by covariance, all its output entropies equal $S_{\min}(\PhiRh)$. It therefore attains the upper bound, and
\begin{equation}
 \chi(\PhiRh)=\log D-S_{\min}(\PhiRh).
 \label{eq:covariantformula}
\end{equation}
The same argument with the product Weyl group gives the corresponding identity for every tensor power. Equations \eqref{eq:exactholevo} and \eqref{eq:exactcapacity} follow.

The explicit local protocol sends $j\in\Z_D$ uniformly as $\proj j$ and measures in the computational basis. It produces the symmetric classical channel $y=j+a$ with $a$ uniform on $R$, whose mutual information is $\log(D/r)$. Since this equals the unrestricted regularized capacity, neither Wigner negativity, entanglement across channel uses, nor collective quantum decoding improves the rate.

Conversely, $\gamma_p=1/r$ is equivalent to $\kappa(p)=\alpha_R$. Equations \eqref{eq:minimax} and \eqref{eq:exactkappa} force $m_R=\alpha_R$ and $\beta_h\leq\alpha_R$; Lemma~\ref{lem:difference-set} then forces $R$ to be a difference set.
\end{proof}

Equations~\eqref{eq:exactentropy} and~\eqref{eq:exactholevo} are fixed-channel tensor-power identities, often called weak additivity. They do not prove $S_{\alpha,\min}(\PhiRh\otimes\Psi)=S_{\alpha,\min}(\PhiRh)+S_{\alpha,\min}(\Psi)$ or $\chi(\PhiRh\otimes\Psi)=\chi(\PhiRh)+\chi(\Psi)$ for an arbitrary second channel $\Psi$. The converse concerns saturation of this collision certificate within uniform $r$-shift factorized models. It neither classifies all random-Weyl channels whose capacity single-letterizes nor rules out other mechanisms outside the stated region.

\section{One-shot converse and strong converse}

The collision certificate also controls finite-blocklength communication. An $(n,M)$ classical code for a channel $\Phi$ consists of states $\rho_1,\ldots,\rho_M$ on $n$ input systems and a decoding positive-operator-valued measure (POVM) $\{\Lambda_m\}_{m=1}^M$ on the $n$ output systems. The codewords may be entangled across all channel uses and the POVM may be collective. For equiprobable messages, define the optimal success probability
\begin{equation}
 \Psucc^{(n)}(M;\Phi):=
 \sup_{\{\rho_m,\Lambda_m\}}
 \frac1M\sum_{m=1}^M
 \Tr\!\left[\Lambda_m\Phi^{\otimes n}(\rho_m)\right],
 \label{eq:successdefinition}
\end{equation}
where the supremum ranges over all such codewords and decoding POVMs.

\begin{theorem}[Nonasymptotic coding bound and strong converse]
\label{thm:strong-converse}
Let $\PhiP$ be an odd-dimensional random-Weyl channel and let $\gamma_p$ be defined in Eq.~\eqref{eq:mainbound}. For every $n,M\geq1$ and every $1<\alpha\leq2$,
\begin{equation}
 \Psucc^{(n)}(M;\PhiP)
 \leq
 \min\!\left\{1,
 \left[\frac{(D\gamma_p)^n}{M}\right]^{(\alpha-1)/\alpha}
 \right\}.
 \label{eq:generaloneshot}
\end{equation}
For the difference-set channels of Theorem~\ref{thm:exact-capacity}, $\gamma_p=1/r$ and hence
\begin{equation}
 \Psucc^{(n)}(M;\PhiRh)
 \leq
 \min\!\left\{1,
 \left[\frac{(D/r)^n}{M}\right]^{(\alpha-1)/\alpha}
 \right\}.
 \label{eq:exactoneshot}
\end{equation}
In particular, if $M_n\geq2^{n\mathsf R}$ with $\mathsf R>\log(D/r)$, then
\begin{equation}
 \Psucc^{(n)}(M_n;\PhiRh)
 \leq2^{-\frac n2[\mathsf R-\log(D/r)]},
 \label{eq:strongexponent}
\end{equation}
so the classical capacity obeys a strong converse.
\end{theorem}

\begin{proof}
Fix an arbitrary code in the supremum in Eq.~\eqref{eq:successdefinition}, set $\sigma_m=\PhiP^{\otimes n}(\rho_m)$, and write $s=(\alpha-1)/\alpha$ and $q=\alpha/(\alpha-1)$, so $1/q=s$. Corollary~\ref{cor:renyi-band} gives $S_\alpha(\sigma_m)\geq-n\log\gamma_p$, or equivalently
\begin{equation}
 \|\sigma_m\|_\alpha
 =\bigl(\Tr\sigma_m^\alpha\bigr)^{1/\alpha}
 \leq\gamma_p^{ns}.
 \label{eq:outputalphanorm}
\end{equation}
Schatten H\"older inequality, together with $0\leq\Lambda_m\leq I$, gives
\begin{equation}
 \Tr(\Lambda_m\sigma_m)
 \leq\|\sigma_m\|_\alpha\|\Lambda_m\|_q
 \leq\|\sigma_m\|_\alpha
 \bigl(\Tr\Lambda_m\bigr)^s.
 \label{eq:holdercode}
\end{equation}
Because $0<s<1$, concavity of $x^s$ and $\sum_m\Lambda_m=I$ imply
\begin{equation}
 \sum_{m=1}^M\bigl(\Tr\Lambda_m\bigr)^s
 \leq M^{1-s}\left(\sum_m\Tr\Lambda_m\right)^s
 =M^{1/\alpha}D^{ns}.
 \label{eq:POVMconcavity}
\end{equation}
Here $\sum_m\Tr\Lambda_m=\Tr I=D^n$ because the output space has dimension $D^n$. Substitution into the code's average success probability gives the nontrivial term in Eq.~\eqref{eq:generaloneshot} for that code. Combining it with the trivial bound $\Psucc^{(n)}\leq1$ and taking the supremum proves Eq.~\eqref{eq:generaloneshot}. The exact family has $\gamma_p=1/r$. Taking $\alpha=2$ and $M_n\geq2^{n\mathsf R}$ gives Eq.~\eqref{eq:strongexponent}.
\end{proof}

The theorem makes the no-advantage statement operational above capacity: no choice of Wigner-negative codewords, inter-use entanglement, or collective measurement can prevent the decoding probability from vanishing exponentially once the communication rate exceeds $\log(D/r)$.
General strong-converse theorems already follow for covariant channels from additive minimum-output R\'enyi entropy~\cite{KonigWehner2009}.  The point of Theorem~\ref{thm:strong-converse} is the channel-specific finite-block estimate~\eqref{eq:generaloneshot}, obtained directly for every code and with its constants explicit; we do not claim the existence of a strong converse itself as new.

\section{Robustness under imperfect flattening}

Let
\begin{equation}
 \delta=\max\{m_R,\beta_h\}-\alpha_R\geq0.
 \label{eq:delta}
\end{equation}
Then
\begin{equation}
 \gamma_p=\frac1r+\left(1-\frac1D\right)\delta.
 \label{eq:gammadelta}
\end{equation}
The local basis code still achieves $\log(D/r)$. For each tensor power, Weyl covariance and Corollary~\ref{cor:renyi-band} give
\begin{equation}
 \chi(\PhiRh^{\otimes n})
 =n\log D-S_{\min}(\PhiRh^{\otimes n})
 \leq n(\log D+\log\gamma_p).
 \label{eq:robustholevo}
\end{equation}
Dividing by $n$, taking the regularized limit, and using Eq.~\eqref{eq:gammadelta} yield
\begin{align}
 \log(D/r)&\leq C(\PhiRh)\leq\log D+\log\gamma_p,
 \label{eq:capacitysandwich}\\
 0&\leq C(\PhiRh)-\log(D/r)
 \leq\log\!\left[1+r\left(1-\frac1D\right)\delta\right].
 \label{eq:robustness}
\end{align}
The right-hand side quantifies the largest possible gain of arbitrary negative and entangled coding over the simple local protocol when the collision spectrum is only approximately flat. It is a one-sided certificate derived from the purity bound; no claim is made that it is tight for a generic nonextremal kernel.

\section{All-use minimum-output and ensemble rigidity}

Assume $1<r<D$, $R$ is a difference set, and
\begin{equation}
 \beta_h<\alpha_R.
 \label{eq:strict}
\end{equation}
Then $0<\kappa(p)=\alpha_R<1$; in this section write $\kappa:=\kappa(p)$. The only nonidentity single-site Weyl modes with squared multiplier $\kappa$ are $W_{0,v}=Z^v$, $v\neq0$.

\begin{proposition}[Rigidity at every tensor power]
\label{prop:rigidity}
Under \eqref{eq:strict}, let $n\geq1$ and $0\leq\alpha\leq2$. If a state $\rho$ satisfies
\begin{equation}
 S_\alpha\!\left(\PhiRh^{\otimes n}(\rho)\right)=n\log r,
 \label{eq:rigidityassumption}
\end{equation}
then
\begin{equation}
 \rho=\proj{j_1}\otimes\cdots\otimes\proj{j_n}
 \label{eq:rigidityconclusion}
\end{equation}
for some $j_1,\ldots,j_n\in\Z_D$, equivalently a product computational-basis stabilizer state up to local Weyl displacements.
\end{proposition}

\begin{proof}
Because $S_\alpha\geq S_2\geq n\log r$, Eq.~\eqref{eq:rigidityassumption} forces equality in the output-purity bound,
\begin{equation}
 \Tr\!\left[\PhiRh^{\otimes n}(\rho)^2\right]=r^{-n}=\gamma_p^n.
 \label{eq:rigiditypurity}
\end{equation}
Equality must therefore hold in both inequalities \eqref{eq:purityupper} and \eqref{eq:puritybinomial}. Since $0<\kappa<1$, every coefficient $\kappa^{|S|}(1-\kappa)^{n-|S|}$ in \eqref{eq:puritybinomial} is strictly positive. Hence
\begin{equation}
 \Tr(\rho_S^2)=1
 \label{eq:allmarginalpure}
\end{equation}
for every $S\subseteq[n]$; in particular, $\rho$ itself is pure.

Equality in \eqref{eq:purityupper} can be written as the vanishing of a sum of nonnegative slack terms,
\begin{equation}
 0=D^{-n}\sum_{\bm z}
 \left[\kappa^{|\supp(\bm z)|}
 -\prod_{i=1}^n|\wh p(z_i)|^2\right]
 |\chi_\rho(\bm z)|^2.
 \label{eq:rigidityslack}
\end{equation}
Consequently, every nonzero characteristic coefficient must use only local nonidentity modes whose squared multiplier is $\kappa$. Under \eqref{eq:strict}, these are precisely the $Z^v$ modes with $v\neq0$; the identity mode is also allowed. Thus
\begin{equation}
 \chi_\rho(\bm z)=0
 \quad\text{whenever any local component }z_i=(u_i,v_i)\text{ has }u_i\neq0.
 \label{eq:onlyZsupport}
\end{equation}
Writing $Z_i$ for $Z$ acting on subsystem $i$, the Weyl expansion of $\rho$ therefore belongs to the abelian algebra generated by $Z_1,\ldots,Z_n$ and is diagonal in the product computational basis. A diagonal density operator with unit purity has exactly one nonzero diagonal entry, proving \eqref{eq:rigidityconclusion}.
\end{proof}

The strict gap is essential for uniqueness. At $\beta_h=\alpha_R$, additional phase-space directions can share the top collision eigenvalue, and further minimizers may occur. The proposition strengthens rate optimality: inter-use entanglement and Wigner negativity are not merely unnecessary for attaining capacity; neither can occur in any minimum-output-R\'enyi minimizer throughout the proved tensor-stable band.

For each restriction below, let the corresponding $C_X$ denote the supremum of asymptotic communication rates achievable with vanishing average error over arbitrary block lengths. In $C_{\mathrm{basis+local}}$, codewords are products of computational-basis states and the receiver performs the product computational-basis measurement followed by arbitrary classical postprocessing. In $C_{\mathrm{stabilizer}}$, every block codeword belongs to the stabilizer polytope (the convex hull of pure stabilizer states) and the decoder is unrestricted. In $C_{\mathrm{Wigner+}}$, arbitrary globally Wigner-nonnegative block codewords are allowed with unrestricted collective decoding. Finally, $C_{\mathrm{unrestricted}}:=C(\PhiRh)$ allows arbitrary codewords and decoders.

\begin{corollary}[Complete optimal ensembles and resource hierarchy]
\label{cor:optimal-ensembles}
Assume the strict-gap conditions of Proposition~\ref{prop:rigidity}. An ensemble $\{p_x,\rho_x\}$ attains $\chi(\PhiRh^{\otimes n})=n\log(D/r)$ if and only if every positive-probability signal $\rho_x$ is a product computational-basis projector and the total probability assigned to each basis string is $D^{-n}$, up to relabeling or splitting identical signals. Consequently,
\begin{equation}
 C_{\mathrm{basis+local}}
 =C_{\mathrm{stabilizer}}
 =C_{\mathrm{Wigner+}}
 =C_{\mathrm{unrestricted}}
 =\log(D/r).
 \label{eq:resourcehierarchy}
\end{equation}
\end{corollary}

\begin{proof}
For any ensemble at block length $n$,
\begin{align}
 \chi
 &=S\!\left(\sum_xp_x\PhiRh^{\otimes n}(\rho_x)\right)
 -\sum_xp_xS\!\left(\PhiRh^{\otimes n}(\rho_x)\right)\notag\\
 &\leq n\log D-n\log r.
 \label{eq:holevoequalitychain}
\end{align}
Equality requires the average output to be $I/D^n$ and every signal state to minimize the output entropy. Proposition~\ref{prop:rigidity} forces all signal states to be computational-basis product projectors.

Let $w(\bm j)$ be the aggregate probability assigned to $\proj{j_1}\otimes\cdots\otimes\proj{j_n}$. On these states the channel is the classical convolution $w\mapsto w*q_R^{\otimes n}$. A difference set satisfies $|\wh q_R(k)|^2=\alpha_R>0$ for every $k\neq0$, while $\wh q_R(0)=1$. The convolution is therefore invertible on functions on $\Z_D^n$. Its output is uniform if and only if $w$ is uniform, which proves the ensemble classification.

The basis ensemble with basis measurement realizes the memoryless classical-noise channel $Y=U+A$, where $U$ is uniform on $\Z_D$ and independent of the noise $A$, which is uniform on $R$. It achieves $I(U;Y)=\log(D/r)$ per use. The restrictions defining the four capacities are nested, while the unrestricted capacity has the same value by Theorem~\ref{thm:exact-capacity}; hence all inequalities collapse to Eq.~\eqref{eq:resourcehierarchy}.
\end{proof}

\section{Two-axis degeneracy and complete endpoint rigidity}
\label{sec:two-axis}

The strict-gap result leaves open what happens when two complementary Weyl directions share the top collision eigenvalue.  The balanced difference-set endpoint permits a complete answer.  We first isolate the elementary pure-state fact that resolves the degeneracy.

\begin{lemma}[Two-axis pure-state lemma]
\label{lem:two-axis}
Let $D$ be odd and let $\rho=\proj\psi$ be a pure state.  If
\begin{equation}
 \chi_\rho(u,v)=0\qquad\text{whenever }u\neq0\text{ and }v\neq0,
 \label{eq:twoaxissupport}
\end{equation}
then $\ket\psi$ is either a computational-basis vector $\ket j$ or a Fourier-basis vector $\ket{\widetilde j}$, up to an overall phase.
\end{lemma}

\begin{proof}
Write $\psi_j=\langle j|\psi\rangle$.  In the Weyl expansion of $\rho$, an off-diagonal matrix element with row--column difference $t\neq0$ receives contributions from the modes $(t,v)$.  Assumption~\eqref{eq:twoaxissupport} leaves only $v=0$, so
\begin{equation}
 \psi_{k+t}\overline{\psi_k}=f(t)
 \label{eq:offdiagonal-circulant}
\end{equation}
for a function $f$ independent of $k$.  If every $f(t\neq0)$ vanishes, the rank-one matrix $\rho$ is diagonal and hence $\ket\psi=\ket j$ for some $j$.

Otherwise $f(t)\neq0$ for some $t$, and Eq.~\eqref{eq:offdiagonal-circulant} shows that every component $\psi_k$ is nonzero.  In particular, $f(1)\neq0$ and $|\psi_{k+1}||\psi_k|$ is independent of $k$.  Going around the odd cycle forces all $|\psi_k|$ to be equal, hence $|\psi_k|=D^{-1/2}$.  Equation~\eqref{eq:offdiagonal-circulant} with $t=1$ then makes the phase increment $\psi_{k+1}/\psi_k$ constant.  Periodicity quantizes it to a $D$th root of unity, giving one of the Fourier vectors in Eq.~\eqref{eq:fourierbasis}.
\end{proof}

Let $R,H\subset\Z_D$ be cyclic difference sets with the same parameters $(D,r,\lambda)$, where $1<r<D$; thus $q_H=\mathbf1_H/r$. Define the balanced channel
\begin{equation}
 \Phi_{R,H,0}:=\Phi_{R,q_H}.
 \label{eq:balancedendpoint}
\end{equation}
Both difference sets have the same nontrivial collision value $\alpha_R=(D-r)/[r(D-1)]$.  The only nonidentity Weyl modes with squared multiplier $\alpha_R$ are the two axes $(0,v\neq0)$ and $(u\neq0,0)$; a mode with both coordinates nonzero has squared multiplier $\alpha_R^2<\alpha_R$.

\begin{theorem}[Complete all-use rigidity at the balanced endpoint]
\label{thm:endpoint-rigidity}
For every $n\geq1$ and $0\leq\alpha\leq2$,
\begin{equation}
 S_{\alpha,\min}(\Phi_{R,H,0}^{\otimes n})=n\log r.
 \label{eq:endpoint-entropy}
\end{equation}
A state attains this minimum if and only if it is a product
\begin{equation}
 \rho=\bigotimes_{i=1}^n\proj{\psi_i},
 \qquad
 \ket{\psi_i}\in\{\ket j:j\in\Z_D\}
 \cup\{\ket{\widetilde j}:j\in\Z_D\},
 \label{eq:endpoint-minimizers}
\end{equation}
where the computational or Fourier basis may be chosen independently at every site.
\end{theorem}

\begin{proof}
Theorem~\ref{thm:exact-capacity} gives Eq.~\eqref{eq:endpoint-entropy}.  Every state in Eq.~\eqref{eq:endpoint-minimizers} attains it: a computational-basis vector is spread uniformly over $r$ orthogonal shifts indexed by $R$, while a Fourier-basis vector is spread uniformly over $r$ orthogonal phase shifts indexed by $H$.  Their tensor products therefore produce normalized rank-$r^n$ projectors.

Conversely, equality at any order $0\leq\alpha\leq2$ forces equality in the output-purity bound exactly as in Proposition~\ref{prop:rigidity}.  Since $0<\alpha_R<1$, equality in Eq.~\eqref{eq:puritybinomial} makes every reduced state pure, so the input is a product of local pure states.  Equality in Eq.~\eqref{eq:purityupper} excludes every local Weyl mode with both coordinates nonzero.  Each local factor therefore satisfies Lemma~\ref{lem:two-axis}, which proves Eq.~\eqref{eq:endpoint-minimizers}.
\end{proof}

\begin{corollary}[All optimal ensembles at the balanced endpoint]
\label{cor:endpoint-ensembles}
An ensemble attains $\chi(\Phi_{R,H,0}^{\otimes n})=n\log(D/r)$ if and only if every positive-probability signal has the form~\eqref{eq:endpoint-minimizers} and its average input state is $I/D^n$.  In particular, any uniform product ensemble built from a fixed sitewise choice of the computational or Fourier basis is optimal.
\end{corollary}

\begin{proof}
Equality in the Holevo bound requires every signal to minimize the output entropy and the average output to be maximally mixed.  Theorem~\ref{thm:endpoint-rigidity} gives the first condition.  Moreover,
\begin{equation}
 \wh p(u,v)=\wh q_H(u)\wh q_R(-v)
\end{equation}
is nonzero at every phase-space point because both difference sets have strictly positive nontrivial Fourier magnitude.  The channel is therefore invertible as a linear map on operators.  Since it is unital, an average output equals $I/D^n$ if and only if the average input does.
\end{proof}

\section{Exact entanglement boundary and entanglement assistance}

Let
\begin{equation}
 \ket{\Omega_D}:=\frac1{\sqrt D}\sum_{j\in\Z_D}\ket j\ket j,
 \qquad \ket{\Omega_{a,b}}:=(I\otimes W_{a,b})\ket{\Omega_D}.
 \label{eq:bell}
\end{equation}
The normalized Choi state is
\begin{equation}
 J_{R,h}=\sum_{a,b}q_R(a)h(b)\proj{\Omega_{a,b}}.
 \label{eq:choi}
\end{equation}
Each Weyl--Bell vector is an odd-dimensional stabilizer state and has a nonnegative discrete Wigner function~\cite{Gross2006}. Thus $J_{R,h}$ is Wigner nonnegative for every $h$, while the channel itself has the positive transition kernel \eqref{eq:wignerkernel} even after adjoining a reference. In what follows, $\Gamma$ denotes partial transpose on the channel-output subsystem and $\lambda_{\min}$ denotes the smallest eigenvalue.

The following criterion resolves the partial-transpose witness frequency by frequency.  It will be used both for the original dephasing path and for the balanced interpolation introduced below.

\begin{theorem}[Fourier-resolved NPT criterion]
\label{thm:fourier-npt}
Let $R$ be a nonempty proper subset of $\Z_D$.  If there are $a\in R$ and $t\neq0$ such that
\begin{equation}
 a+t\notin R,
 \qquad \wh h(t)\neq0,
 \label{eq:nptcriterion}
\end{equation}
then $J_{R,h}$ is NPT.  More precisely,
\begin{equation}
 \lambda_{\min}(J_{R,h}^{\Gamma})
 \leq
 \frac{1-\sqrt{1+4|\wh h(t)|^2}}{2rD}<0.
 \label{eq:generalNPTbound}
\end{equation}
\end{theorem}

\begin{proof}
Expanding Eq.~\eqref{eq:choi} and partially transposing the output gives
\begin{equation}
 J_{R,h}^{\Gamma}
 =\frac1{rD}\sum_{a\in R}\sum_{j,k\in\Z_D}
 \wh h(j-k)\,|j,k+a\rangle\!\langle k,j+a|.
 \label{eq:generalChoiPT}
\end{equation}
Fix the pair in Eq.~\eqref{eq:nptcriterion} and any $j$.  On the ordered computational-basis vectors
\begin{equation}
 |j,j+a+t\rangle,
 \qquad |j+t,j+a\rangle,
\end{equation}
the corresponding principal submatrix is
\begin{equation}
 \frac1{rD}
 \begin{pmatrix}
  0&\wh h(-t)\\
  \wh h(t)&s
 \end{pmatrix},
 \qquad s:=\mathbf1_R(a-t)\in\{0,1\}.
 \label{eq:generalNPTminor}
\end{equation}
The zero in the first diagonal entry follows from $a+t\notin R$, and the selected off-diagonal entry has a unique contribution in Eq.~\eqref{eq:generalChoiPT}. If $s=1$, the smaller eigenvalue is the right-hand side of Eq.~\eqref{eq:generalNPTbound}. If $s=0$, it is $-|\wh h(t)|/(rD)$; because $|\wh h(t)|\leq1$, this is no larger than the right-hand side of Eq.~\eqref{eq:generalNPTbound}. Eigenvalue interlacing proves the claim.
\end{proof}

We now specialize to a phase profile that interpolates between complete dephasing and the identity. Recall that $u_D(b)=1/D$, and let $\delta_0(b):=\delta_{b,0}$ denote the point mass at zero. Define
\begin{equation}
 h_\eta=\eta\delta_0+(1-\eta)u_D,
 \qquad 0\leq\eta\leq1.
 \label{eq:heta}
\end{equation}
Write $\PhiReta:=\Phi_{R,h_\eta}$ and $J_{R,\eta}:=J_{R,h_\eta}$ for the corresponding channel and normalized Choi state.
Writing $\Delta_Z(\rho)=\sum_j\langle j|\rho|j\rangle\proj j$, the channel becomes
\begin{equation}
 \PhiReta(\rho)=\frac1r\sum_{a\in R}X^a
 \left[\eta\rho+(1-\eta)\Delta_Z(\rho)\right]X^{-a}.
 \label{eq:dephasingform}
\end{equation}

\begin{theorem}[Exact entanglement-breaking boundary]
\label{thm:eb-boundary}
Let $R$ be any nonempty proper subset of $\Z_D$. Then $\PhiReta$ is entanglement breaking if and only if $\eta=0$. For every $\eta>0$, the Choi state is negative under partial transpose (NPT). More quantitatively,
\begin{equation}
 \lambda_{\min}\!\left(J_{R,\eta}^{\Gamma}\right)
 \leq\frac{1-\sqrt{1+4\eta^2}}{2rD}<0,
 \label{eq:NPTbound}
\end{equation}
\end{theorem}

\begin{proof}
At $\eta=0$, Eq.~\eqref{eq:dephasingform} measures in the computational basis and prepares the classical mixture
\begin{equation}
 \Phi_{R,0}(\rho)=\sum_j\langle j|\rho|j\rangle
 \left(\frac1r\sum_{a\in R}\proj{j+a}\right),
 \label{eq:measureprepare}
\end{equation}
so the channel is entanglement breaking.

For $\eta>0$, direct expansion of the Choi state gives
\begin{align}
 J_{R,\eta}^{\Gamma}
 &=\frac1{rD}\sum_{a\in R}\sum_j
 \proj{j,j+a}\notag\\
 &\quad+\frac{\eta}{rD}\sum_{a\in R}\sum_{j\neq k}
 |j,k+a\rangle\!\langle k,j+a|.
 \label{eq:ChoiPT}
\end{align}
Choose $a\in R$ and $b\notin R$, set $t=b-a\neq0$, and take $k=j+t$. The computational-basis vector $|j,j+b\rangle=|j,k+a\rangle$ has zero diagonal entry in $J_{R,\eta}^{\Gamma}$ because $b\notin R$, but it has off-diagonal matrix element $\eta/(rD)$ with $|k,j+a\rangle$. The ordered indices in Eq.~\eqref{eq:ChoiPT} show that this selected off-diagonal entry has a unique contribution, so its coefficient cannot cancel. A positive semidefinite matrix cannot have a nonzero entry in a row whose diagonal entry is zero. Hence $J_{R,\eta}^{\Gamma}$ is not positive.

For the quantitative bound, the principal submatrix on these two basis vectors is
\begin{equation}
 \frac1{rD}\begin{pmatrix}0&\eta\\ \eta&s\end{pmatrix},
 \qquad s:=\mathbf1_R(2a-b)\in\{0,1\}.
 \label{eq:NPTminor}
\end{equation}
Its smallest eigenvalue is $-\eta/(rD)$ if $s=0$ and $[1-\sqrt{1+4\eta^2}]/(2rD)$ if $s=1$. In the first case, $\sqrt{1+4\eta^2}\leq1+2\eta$ shows that $-\eta/(rD)$ is no larger than the right-hand side of Eq.~\eqref{eq:NPTbound}. Eigenvalue interlacing therefore proves that bound in both cases. An entanglement-breaking channel has a separable Choi state~\cite{HorodeckiShorRuskai2003}, and separability implies positive partial transpose~\cite{Peres1996,Horodecki1996}; hence every $\eta>0$ channel is not entanglement breaking.
\end{proof}

The theorem is stronger than a Bell-overlap witness: it determines the exact boundary for the full interpolation and does not use the difference-set property. It does not imply positive quantum capacity; NPT is used only to certify persistent channel entanglement.

For a probability distribution $f$, write $H(f)=-\sum_xf(x)\log f(x)$, with the convention $0\log0:=0$. Pre-shared entanglement produces a complementary operational separation.

\begin{proposition}[Entanglement-assisted capacity]
\label{prop:ea-capacity}
For every factorized channel in Eq.~\eqref{eq:productchannel},
\begin{equation}
 \CE(\PhiRh)=2\log D-\log r-H(h).
 \label{eq:EAcapacity}
\end{equation}
If $R$ is a difference set, $\beta_{h_\eta}\leq\alpha_R$, and $\eta>0$, then
\begin{equation}
 \CE(\PhiReta)-C(\PhiReta)=\log D-H(h_\eta)>0.
 \label{eq:EAgap}
\end{equation}
\end{proposition}

\begin{proof}
Let $\psi^\rho_{A'A}$ be the density operator of a purification of the input state $\rho_A$, where $A'$ is the purifying reference, and let $\operatorname{id}_{A'}$ denote the identity channel on that reference. Define the channel mutual information by
\begin{equation}
 I(\rho,\Phi):=S(\rho)+S(\Phi(\rho))
 -S\!\left[(\operatorname{id}_{A'}\otimes\Phi)(\psi^\rho_{A'A})\right].
 \label{eq:channelmutualinformation}
\end{equation}
The entanglement-assisted capacity theorem states that $\CE(\Phi)=\max_\rho I(\rho,\Phi)$~\cite{Bennett1999,Bennett2002}. Concavity of $I(\rho,\PhiRh)$ in $\rho$, Weyl covariance, and invariance under the corresponding input and output unitaries give
\begin{equation}
 I(I/D,\PhiRh)
 \geq\frac1{D^2}\sum_{c,d\in\Z_D}I(W_{c,d}\rho W_{c,d}^{\dagger},\PhiRh)
 =I(\rho,\PhiRh)
 \label{eq:EAtwirl}
\end{equation}
for every $\rho$, where Eq.~\eqref{eq:weyltwirl} was used. Thus the maximally mixed input is optimal. Its input and output entropies are both $\log D$, while the joint reference--output state is the Bell-diagonal Choi state \eqref{eq:choi}. Orthogonality of the Weyl--Bell basis gives
\begin{equation}
 S(J_{R,h})=H(q_R\otimes h)=H(q_R)+H(h)=\log r+H(h),
\end{equation}
which proves Eq.~\eqref{eq:EAcapacity}. Theorem~\ref{thm:exact-capacity} gives $C=\log(D/r)$ in the stated region. Finally, $h_\eta$ is nonuniform for every $\eta>0$, so $H(h_\eta)<\log D$.
\end{proof}

\section{Balanced bi-difference-set interpolation}
\label{sec:balanced-family}

Let $R,H\subset\Z_D$ be cyclic difference sets with the same parameters $(D,r,\lambda)$ and $1<r<D$.  With $q_H=\mathbf1_H/r$, define
\begin{equation}
 h_\varepsilon=(1-\varepsilon)q_H+\varepsilon u_D,
 \qquad
 \PhiRHeps:=\Phi_{R,h_\varepsilon},
 \qquad 0\leq\varepsilon\leq1.
 \label{eq:balancedfamily}
\end{equation}
This path removes the phase difference-set component without changing the uniform shift support.  Its two collision axes are degenerate only at $\varepsilon=0$.

\begin{theorem}[Constant capacity across an exact entanglement boundary]
\label{thm:balanced-family}
For every $n\geq1$, $0\leq\alpha\leq2$, and $0\leq\varepsilon\leq1$,
\begin{align}
 S_{\alpha,\min}(\PhiRHeps^{\otimes n})&=n\log r,
 \label{eq:balancedentropy}\\
 \chi(\PhiRHeps^{\otimes n})&=n\log(D/r),
 \qquad C(\PhiRHeps)=\log(D/r).
 \label{eq:balancedcapacity}
\end{align}
The nonasymptotic bound~\eqref{eq:exactoneshot} and its strong converse hold throughout the interval.  The complete minimizer and optimal-ensemble classifications are as follows.
\begin{enumerate}
 \item At $\varepsilon=0$, the minimizers are exactly the sitewise computational/Fourier products in Eq.~\eqref{eq:endpoint-minimizers}; an ensemble is optimal exactly under the conditions of Corollary~\ref{cor:endpoint-ensembles}.
 \item For $0<\varepsilon\leq1$, the minimizers are exactly computational-basis products, and every optimal ensemble has uniform aggregate weight $D^{-n}$ on the $D^n$ basis strings.
\end{enumerate}
Moreover, $\PhiRHeps$ is entanglement breaking if and only if $\varepsilon=1$, and its Choi state is NPT for every $0\leq\varepsilon<1$.  Its entanglement-assisted capacity is
\begin{equation}
 \CE(\PhiRHeps)=2\log D-\log r-H(h_\varepsilon),
 \label{eq:balancedEA}
\end{equation}
so that
\begin{equation}
 \CE(\Phi_{R,H,0})=2C(\Phi_{R,H,0}),
 \qquad
 \CE(\Phi_{R,H,1})=C(\Phi_{R,H,1}).
 \label{eq:balancedendpoints}
\end{equation}
In particular, $\CE>C$ at every NPT point of the path.
\end{theorem}

\begin{proof}
For $k\neq0$, difference-set flatness and the vanishing nonconstant Fourier modes of $u_D$ give
\begin{equation}
 |\wh h_\varepsilon(k)|^2
 =(1-\varepsilon)^2|\wh q_H(k)|^2
 =(1-\varepsilon)^2\alpha_R.
 \label{eq:balancedFourier}
\end{equation}
Thus $\beta_{h_\varepsilon}\leq\alpha_R$, with equality only at $\varepsilon=0$.  Theorem~\ref{thm:exact-capacity} and Theorem~\ref{thm:strong-converse} prove Eqs.~\eqref{eq:balancedentropy}--\eqref{eq:balancedcapacity} and the coding bound.  The equality case at $\varepsilon=0$ is Theorem~\ref{thm:endpoint-rigidity} and Corollary~\ref{cor:endpoint-ensembles}; for $\varepsilon>0$, Proposition~\ref{prop:rigidity} and Corollary~\ref{cor:optimal-ensembles} apply because the gap is strict.

At $\varepsilon=1$, $h_\varepsilon=u_D$ and Eq.~\eqref{eq:measureprepare} shows that the channel is entanglement breaking.  If $\varepsilon<1$, choose $a\in R$ and $b\notin R$ and put $t=b-a\neq0$.  Equation~\eqref{eq:balancedFourier} implies $\wh h_\varepsilon(t)\neq0$, so Theorem~\ref{thm:fourier-npt} makes the Choi state NPT.  Finally, Eq.~\eqref{eq:balancedEA} is Proposition~\ref{prop:ea-capacity}.  The endpoint identities use $H(q_H)=\log r$ and $H(u_D)=\log D$; strict concavity of entropy gives $H(h_\varepsilon)<\log D$ for $\varepsilon<1$.
\end{proof}

The exact equality of the unassisted capacity at the two endpoints is therefore not a continuity statement: the same rate and strong converse hold at every intermediate point, even though the Choi state crosses from NPT to separable only at the final endpoint and the optimizer set changes discontinuously at the initial endpoint.

\section{Cyclic Singer family and scalable separation}

We recall the projective-geometric origin of the family. For a prime power $q$, let $\mathrm{PG}(2,q)$ be the projective plane of order $q$. Singer proved that it admits a cyclic collineation group $\langle\sigma\rangle$ of order $D=q^2+q+1$ acting regularly on its points (and hence on its lines)~\cite{Singer1938,Beth1999}. Fix a point $P$ and a line $\ell$, identify $\langle\sigma\rangle$ with $\mathbb Z_D$, and define
\begin{equation}
 R=\{t\in\mathbb Z_D:\ \sigma^t(P)\in\ell\}.
 \label{eq:singerconstruction}
\end{equation}
The line $\ell$ contains $q+1$ points, so $|R|=q+1$. For each nonzero $s\in\mathbb Z_D$, the lines $\ell$ and $\sigma^{-s}(\ell)$ are distinct and intersect in exactly one point. Writing that point uniquely as $\sigma^{a'}(P)$ gives a unique ordered pair $a,a'\in R$ with $a-a'=s$. Hence $R$ is a cyclic Singer difference set with
\begin{equation}
 (D,r,\lambda)=(q^2+q+1,q+1,1).
 \label{eq:singer}
\end{equation}
Because $q$ and $q+1$ have opposite parity, $D=q(q+1)+1$ is odd for every prime power $q$; no even-$q$ restriction is needed. Since $D=r(r-1)+1$,
\begin{equation}
 \alpha_R=\frac{D-r}{r(D-1)}=\frac{r-1}{r^2}.
 \label{eq:alphasinger}
\end{equation}
Choose two Singer difference sets $R,H\subset\Z_D$ with these parameters; one may simply take $H=R$.  For the balanced profile~\eqref{eq:balancedfamily}, put
\begin{equation}
 a_\varepsilon:=\frac{1-\varepsilon}{r}+\frac{\varepsilon}{D},
 \qquad
 b_\varepsilon:=\frac{\varepsilon}{D},
 \label{eq:balancedmasses}
\end{equation}
so that $h_\varepsilon$ takes the value $a_\varepsilon$ on $H$ and $b_\varepsilon$ on its complement.  Define
\begin{equation}
 p_H:=ra_\varepsilon=1-\varepsilon+\varepsilon r/D,
 \qquad
 p_{H^c}:=(D-r)b_\varepsilon=\varepsilon(1-r/D).
 \label{eq:groupmasses}
\end{equation}

\begin{corollary}[Scalable Singer separation]
\label{cor:singer-balanced}
For every prime power $q$ and every $0\leq\varepsilon\leq1$, the Singer channel $\PhiRHeps$ has all the properties in Theorem~\ref{thm:balanced-family}.  Its phase entropy is exactly
\begin{equation}
 H(h_\varepsilon)
 =H_2(p_H)+p_H\log r+p_{H^c}\log(D-r),
 \label{eq:balancedentropyformula}
\end{equation}
where $H_2(x)=-x\log x-(1-x)\log(1-x)$.  As $q\to\infty$ through prime powers, at every fixed $\varepsilon$,
\begin{align}
 \CE(\PhiRHeps)
 &=(2-\varepsilon)C(\PhiRHeps)-H_2(\varepsilon)+o(1),
 \label{eq:EAasymptotic}\\
 \lim_{q\to\infty}\frac{\CE(\PhiRHeps)}{C(\PhiRHeps)}
 &=2-\varepsilon.
 \label{eq:EAratio}
\end{align}
Consequently, for each fixed $0\leq\varepsilon<1$, the additive rate advantage $\CE-C$ diverges as $q\to\infty$, while the unrestricted unassisted capacity remains achievable by the same local basis protocol.
\end{corollary}

\begin{proof}
Equation~\eqref{eq:balancedentropyformula} is the entropy chain rule for the partition $H\cup H^c$, within whose two cells $h_\varepsilon$ is uniform.  Since $r/D\to0$, $\log(D-r)=\log D+o(1)$, and $C=\log(D/r)\sim\log q$, it gives
\begin{equation}
 H(h_\varepsilon)
 =H_2(\varepsilon)+(1-\varepsilon)\log r
 +\varepsilon\log D+o(1).
 \label{eq:phaseentropyasymptotic}
\end{equation}
Substitution into Eq.~\eqref{eq:balancedEA} proves Eq.~\eqref{eq:EAasymptotic}; division by the diverging $C$ proves Eq.~\eqref{eq:EAratio}.  The remaining claims follow from Theorem~\ref{thm:balanced-family}.
\end{proof}

\section{Identity--dephasing path: exact collision phase diagram and capacity conjecture}
\label{sec:identity-dephasing}

For comparison with the balanced path, return to the identity--dephasing profile~\eqref{eq:heta}.  The exact unassisted capacity is certified only below a threshold, but its minimum collision entropy and all optimizers can be solved over the entire interval.

For a general difference set, define the exact-capacity threshold
\begin{equation}
 \eta_\star:=\sqrt{\alpha_R}
 =\sqrt{\frac{D-r}{r(D-1)}};
 \qquad
 \eta_\star=\frac{\sqrt{r-1}}r\quad\text{for Singer parameters}.
 \label{eq:etastar}
\end{equation}
For the phase profile $h_\eta$ in Eq.~\eqref{eq:heta},
\begin{equation}
 h_\eta(0)=\frac{1+(D-1)\eta}{D},\qquad
 h_\eta(b\neq0)=\frac{1-\eta}{D},\qquad
 \wh h_\eta(u\neq0)=\eta.
 \label{eq:hetavalues}
\end{equation}

\begin{theorem}[All-use collision-entropy bifurcation]
\label{thm:eta-collision-phase}
Let $R$ be any nontrivial cyclic difference set and define
\begin{equation}
 \gamma_\eta:=\frac1D+\left(1-\frac1D\right)
 \max\{\alpha_R,\eta^2\}.
 \label{eq:gammaeta}
\end{equation}
Then, for every $n\geq1$ and $0\leq\eta\leq1$,
\begin{equation}
 S_{2,\min}(\PhiReta^{\otimes n})=-n\log\gamma_\eta.
 \label{eq:fullS2phase}
\end{equation}
Every minimizing input is a product state, with the complete list
\begin{equation}
 \begin{cases}
  \bigotimes_i\proj{j_i}, & 0\leq\eta<\sqrt{\alpha_R},\\[2pt]
  \bigotimes_i\proj{\psi_i},\quad
  \ket{\psi_i}\in\{\ket j\}_j\cup\{\ket{\widetilde j}\}_j,
   & \eta=\sqrt{\alpha_R},\\[2pt]
  \bigotimes_i\proj{\widetilde j_i},
   & \sqrt{\alpha_R}<\eta\leq1.
 \end{cases}
 \label{eq:etaoptimizers}
\end{equation}
Thus the collision optimizer changes sharply from the computational Weyl axis to the Fourier Weyl axis, with independent sitewise choices at the unique degeneracy point.
\end{theorem}

\begin{proof}
Equations~\eqref{eq:exactkappa} and~\eqref{eq:hetavalues} give $\kappa=\max\{\alpha_R,\eta^2\}$, so Theorem~\ref{thm:collision-bound} gives the lower bound in Eq.~\eqref{eq:fullS2phase}.  Below the threshold it is attained by computational-basis products, whose single-site output purity is $1/r= D^{-1}+(1-D^{-1})\alpha_R$.  Above the threshold it is attained by Fourier-basis products: the shift unitaries contribute only phases and the output eigenvalues are $h_\eta$, with
\begin{equation}
 \sum_bh_\eta(b)^2=\frac{1+(D-1)\eta^2}{D}=\gamma_\eta.
 \label{eq:fourierpurity}
\end{equation}
Both constructions attain the same value at the threshold, including arbitrary products that choose either basis at each site.

It remains to prove completeness.  For $\eta<1$, one has $0<\kappa<1$.  Equality in Eqs.~\eqref{eq:purityupper} and~\eqref{eq:puritybinomial} forces every one-site marginal to be pure, hence the input is a product of local pure states.  The slack identity~\eqref{eq:rigidityslack} then restricts every local characteristic function to the top collision modes: the $Z$ axis below the threshold, both axes at equality, and the $X$ axis above it.  Diagonality in the computational or Fourier basis proves the first and third lines of Eq.~\eqref{eq:etaoptimizers}; Lemma~\ref{lem:two-axis} proves the middle line.

At $\eta=1$, Eq.~\eqref{eq:fullS2phase} requires the minimizing output to have purity one. A convex combination of unitary conjugates of a density operator can be rank one only if the input itself is pure and every conjugate equals the same rank-one projector. Hence all shifted copies of the input coincide up to phase. Comparing two shift strings that differ at one site shows that the input is a joint eigenvector of every local $X^{a-a'}$ with $a,a'\in R$. The difference-set property gives $R-R=\Z_D$, so each local factor is an $X$ eigenvector and the joint eigenspaces are precisely the product Fourier vectors. This completes the endpoint case.
\end{proof}

\begin{figure}[H]
 \centering
 \includegraphics[width=0.98\textwidth]{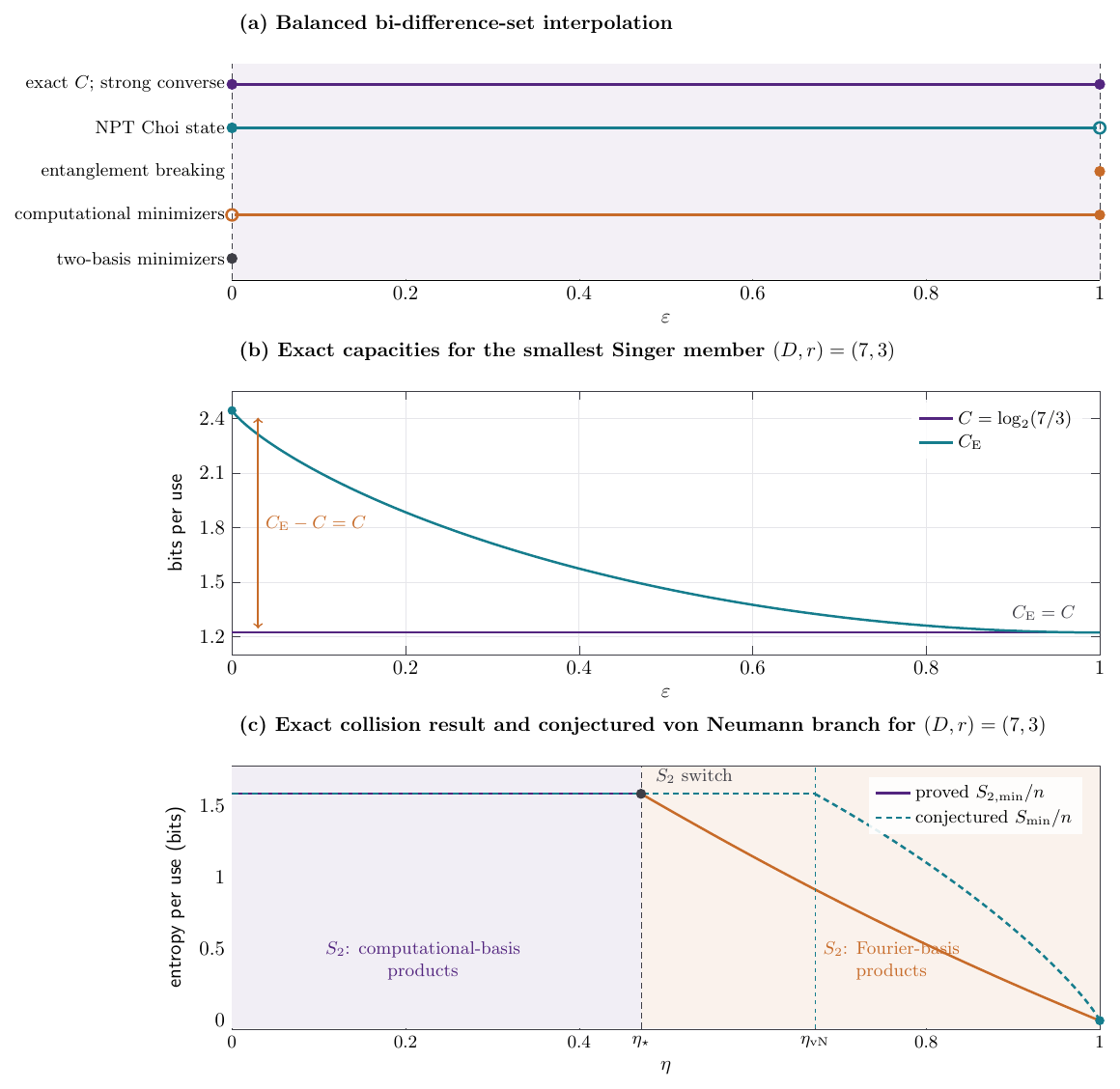}
 \caption{Exact operational and collision-entropy structure. (a) Along the balanced interpolation, $C=\log(D/r)$ and the strong converse hold on the full interval, the Choi state is NPT for $\varepsilon<1$, and the channel is entanglement breaking at $\varepsilon=1$; the optimizer set has a two-basis degeneracy only at $\varepsilon=0$. (b) For the smallest Singer member $(D,r)=(7,3)$, $C$ is constant while $C_{\mathrm E}$ decreases continuously from $2C$ to $C$. (c) For the identity--dephasing path, the solid curve is the proved all-use minimum collision entropy and changes branch at $\eta_\star=\sqrt{r-1}/r$. The dashed curve is the von Neumann expression proposed in Conjecture~\ref{conj:eta-renyi}; its distinct crossover $\eta_{\rm vN}$ is not presented as a proved capacity transition.}
 \label{fig:phase-diagram}
\end{figure}

\subsection{Rigorous capacity bounds from two complementary basis codes}

The collision optimizer need not minimize the von Neumann entropy.  Nevertheless, the two distinguished bases give two rigorous achievable rates, and the collision theorem gives a regularized converse.

\begin{proposition}[All-use capacity sandwich]
\label{prop:eta-capacity-sandwich}
Let $R$ be a nontrivial cyclic difference set and define
\begin{equation}
 L_\eta:=\max\!\left\{\log(D/r),\,\log D-H(h_\eta)\right\},
 \qquad
 U_\eta:=\log D+\log\gamma_\eta.
 \label{eq:etaLU}
\end{equation}
Then, for every $n\geq1$,
\begin{equation}
 nL_\eta\leq \chi(\PhiReta^{\otimes n})\leq nU_\eta,
 \qquad
 L_\eta\leq C(\PhiReta)\leq U_\eta.
 \label{eq:eta-capacity-sandwich}
\end{equation}
Both inequalities coincide for $0\leq\eta\leq\eta_\star$ and at $\eta=1$.
\end{proposition}

\begin{proof}
The uniform computational-basis ensemble followed by the same basis measurement realizes the additive-noise channel with noise $q_R$ and achieves $\log(D/r)$ per use.  For a Fourier projector, Eq.~\eqref{eq:heta} gives
\begin{equation}
 \PhiReta(\proj{\widetilde j})
 =\eta\proj{\widetilde j}+(1-\eta)\frac ID.
 \label{eq:fourieroutput}
\end{equation}
The shift average is immaterial because $\ket{\widetilde j}$ is an $X$ eigenvector.  The output spectrum is $h_\eta$, and the uniformly weighted Fourier ensemble has average output $I/D$; it therefore achieves $\log D-H(h_\eta)$.  Product use of either ensemble proves the lower bound for every $n$.

For the upper bound, Weyl covariance at $n$ uses and Eq.~\eqref{eq:fullS2phase} give
\begin{align}
 \chi(\PhiReta^{\otimes n})
 &=n\log D-S_{\min}(\PhiReta^{\otimes n})\notag\\
 &\leq n\log D-S_{2,\min}(\PhiReta^{\otimes n})
 =n(\log D+\log\gamma_\eta).
\end{align}
Regularization proves the capacity bounds.  If $\eta\leq\eta_\star$, then $\gamma_\eta=1/r$ and both bounds equal $\log(D/r)$.  At $\eta=1$, $H(h_1)=0$ and $\gamma_1=1$, so both equal $\log D$.
\end{proof}

The exact-capacity condition is
\begin{equation}
 \eta^2\leq\alpha_R
 \quad\Longleftrightarrow\quad
 \eta\leq\eta_\star.
 \label{eq:capacitycondition}
\end{equation}
Theorem~\ref{thm:eb-boundary} shows that the channel has an NPT Choi state for every $\eta>0$. Therefore the exact-capacity and non-entanglement-breaking regions overlap on the full interval
\begin{equation}
 \boxed{0<\eta\leq\eta_\star}.
 \label{eq:overlap}
\end{equation}

\begin{corollary}[Identity--dephasing CPWP--NPT subfamily]
\label{cor:singer-separation}
For every prime power $q$, let $D=q^2+q+1$, $r=q+1$, and let $R\subset\Z_D$ be a cyclic Singer difference set. For every $0<\eta\leq\eta_\star=\sqrt{r-1}/r$, the channel $\Phi_{R,\eta}$ is CPWP and has an NPT Choi state, yet
\begin{equation}
 C_{\mathrm{basis+local}}
 =C_{\mathrm{stabilizer}}
 =C_{\mathrm{Wigner+}}
 =C_{\mathrm{unrestricted}}
 =\log(D/r)<\CE(\Phi_{R,\eta}).
 \label{eq:singerheadline}
\end{equation}
The unassisted capacity obeys the nonasymptotic bound in Eq.~\eqref{eq:exactoneshot} and hence a strong converse against arbitrary codes. If $0<\eta<\eta_\star$, every minimum-output-R\'enyi state in the proved band and every Holevo-optimal block ensemble at every tensor power has the product-basis form classified in Proposition~\ref{prop:rigidity} and Corollary~\ref{cor:optimal-ensembles}.
\end{corollary}

\begin{proof}
Singer difference sets satisfy $m_R=\alpha_R=(r-1)/r^2$, while Eq.~\eqref{eq:hetavalues} gives $\beta_{h_\eta}=\eta^2$. Thus $\eta\leq\eta_\star$ is exactly the hypothesis of Theorem~\ref{thm:exact-capacity}, and the strict inequality is exactly the gap used in Proposition~\ref{prop:rigidity}. CPWP follows from the positive Wigner transition kernel in Eq.~\eqref{eq:wignerkernel}; NPT follows from Theorem~\ref{thm:eb-boundary}; the strong converse follows from Theorem~\ref{thm:strong-converse}; and the strict entanglement-assisted inequality follows from Proposition~\ref{prop:ea-capacity}. The local basis protocol attains the unrestricted capacity, so the nested restricted capacities coincide even at the endpoint.
\end{proof}

At $\eta=0$ the capacity formula and rigidity still hold, but the channel is entanglement breaking and $\CE=C$. Theorem~\ref{thm:eta-collision-phase} additionally closes the collision-entropy optimizer problem at and above $\eta_\star$; except for the trivial noiseless endpoint $\eta=1$, it does not determine the von Neumann minimum output entropy or the classical capacity in the interval $\eta_\star<\eta<1$. Figure~\ref{fig:phase-diagram} summarizes both interpolation paths. The Singer construction supplies an infinite scalable sequence for all prime powers $q$, including odd $q$ (for example, $q=3$ gives $D=13$) and dimensions that are not prime powers (for example, $q=4$ gives $D=21$). It is not an existence result for every odd dimension or every pair $(D,r)$.

For the smallest member, $q=2$ gives $(D,r)=(7,3)$. One representative difference set is $R=\{0,1,3\}\subset\mathbb Z_7$. Choosing $\eta=0.4$ yields
\begin{equation}
 C=\log_2(7/3)\approx1.2224\ \text{bits/use},
 \label{eq:D7capacity}
\end{equation}
while the entanglement-assisted capacity and the NPT eigenvalue certificate satisfy
\begin{equation}
 \CE\approx1.7009\ \text{bits/use},
 \qquad
 \lambda_{\min}(J_{R,\eta}^{\Gamma})\leq-0.00668.
 \label{eq:D7extra}
\end{equation}

\subsection{Von Neumann crossover and an open tensor-power conjecture}

For a probability vector $f$, let
\begin{equation}
 H_\alpha^{\rm R}(f):=\frac1{1-\alpha}\log\sum_x f(x)^\alpha
 \quad(\alpha\neq1),
 \qquad H_1^{\rm R}(f):=H(f).
 \label{eq:classical-renyi}
\end{equation}
Because $1<r<D$, there is a unique $\eta_{\rm vN}\in(0,1)$ satisfying
\begin{equation}
 H(h_{\eta_{\rm vN}})=\log r.
 \label{eq:etavn}
\end{equation}
Indeed, $H(h_0)=\log D$, $H(h_1)=0$, and for $0<\eta<1$,
\begin{equation}
 \frac{d}{d\eta}H(h_\eta)
 =-\frac{D-1}{D}\log\!\left(\frac{1+(D-1)\eta}{1-\eta}\right)<0.
 \label{eq:hetaderivative}
\end{equation}
Moreover, $H_2^{\rm R}(h_{\eta_\star})=\log r$ by Eqs.~\eqref{eq:gammaeta} and~\eqref{eq:gammar}. Since $0<\eta_\star<1$, the distribution $h_{\eta_\star}$ has full support and is nonuniform, so strict monotonicity of R\'enyi entropy gives $H(h_{\eta_\star})>H_2^{\rm R}(h_{\eta_\star})$. Hence
\begin{equation}
 \eta_{\rm vN}>\eta_\star.
 \label{eq:thresholdseparation}
\end{equation}
This strict separation rules out an inference of the von Neumann optimizer transition from the proved collision transition.

\begin{conjecture}[Two-basis tensor-power R\'enyi formula]
\label{conj:eta-renyi}
For every nontrivial cyclic difference set $R\subset\Z_D$, every $n\geq1$, every $0\leq\eta\leq1$, and every $1\leq\alpha\leq2$,
\begin{equation}
 S_{\alpha,\min}(\PhiReta^{\otimes n})
 =n\min\!\left\{\log r,H_\alpha^{\rm R}(h_\eta)\right\}.
 \label{eq:eta-renyi-conjecture}
\end{equation}
For $\alpha=1$ this would imply
\begin{equation}
 C(\PhiReta)
 =\max\!\left\{\log(D/r),\log D-H(h_\eta)\right\},
 \label{eq:eta-capacity-conjecture}
\end{equation}
with computational-basis minimizers below $\eta_{\rm vN}$ and Fourier-basis minimizers above it.  At equality both bases attain the same entropy; no assertion about the complete minimizer set at $\eta_{\rm vN}$ is included in the conjecture.
\end{conjecture}

The conjecture is already a theorem at $\alpha=2$ by Theorem~\ref{thm:eta-collision-phase}, throughout $0\leq\eta\leq\eta_\star$ by Theorem~\ref{thm:exact-capacity}, and at the noiseless endpoint $\eta=1$.  Its unresolved content is the von Neumann-to-collision strip $\eta_\star<\eta<1$, $1\leq\alpha<2$, together with tensor stability against arbitrary entangled inputs.

A natural proof target is the maximal output Schatten-norm inequality
\begin{equation}
 \nu_\alpha(\PhiReta^{\otimes n})
 \leq
 \max\!\left\{r^{(1-\alpha)/\alpha},
 \|h_\eta\|_\alpha\right\}^{n},
 \qquad 1<\alpha\leq2,
 \label{eq:desired-pnorm}
\end{equation}
where $\nu_\alpha(\Phi):=\max_\rho\|\Phi(\rho)\|_\alpha$.  Both terms on the right are attained by the corresponding product-basis inputs.  The channel factors into two commuting dephasing operations,
\begin{equation}
 \PhiReta=\mathcal A_R\circ\mathcal D_{Z,\eta}
 =\mathcal D_{Z,\eta}\circ\mathcal A_R,
 \quad
 \mathcal A_R(\rho)=\frac1r\sum_{a\in R}X^a\rho X^{-a},
 \quad
 \mathcal D_{Z,\eta}=\eta\,\operatorname{id}+(1-\eta)\Delta_Z.
 \label{eq:commuting-factorization}
\end{equation}
However, the proof of Theorem~\ref{thm:collision-bound} retains only squared multiplier magnitudes and reduced purities. It therefore proves Eq.~\eqref{eq:desired-pnorm} at $\alpha=2$ but discards the spectral information needed near $\alpha=1$. Interpolation with the trace norm reproduces only the collision lower bound and does not establish Eq.~\eqref{eq:eta-renyi-conjecture}. This is the precise missing step; no later proved result in this paper depends on it. Related work on Pauli-axis and Weyl channels likewise obtains $p=2$ multiplicativity or capacity bounds under additional symmetry or majorization hypotheses, rather than the all-use entropy identity asserted here~\cite{NathansonRuskai2007,Siudzinska2020,Rehman2018}.

As diagnostic evidence for the von Neumann case $\alpha=1$, we performed multi-start pure-state optimization for three small cyclic difference sets. The observed one-use minimum output entropy agreed with $\min\{\log r,H(h_\eta)\}$ at all tested values, and no intermediate optimizer was found. For $D=7$, we also enumerated all $(D+1)(D^2+1)=400$ two-qudit Lagrangian stabilizer subspaces; at tested values spanning both sides of $\eta_{\rm vN}$, no entangled stabilizer input improved on twice the one-use von Neumann candidate. The threshold values used in these checks are listed in Table~\ref{tab:eta-numerics}.
\begin{table}[H]
\centering
\begin{tabular}{@{}cclcc@{}}
\toprule
$D$ & $r$ & representative $R$ & $\eta_\star$ & $\eta_{\rm vN}$ \\
\midrule
7  & 3 & $\{0,1,3\}$       & 0.471405 & 0.671640 \\
11 & 5 & $\{1,3,4,5,9\}$   & 0.346410 & 0.553720 \\
13 & 4 & $\{0,1,3,9\}$     & 0.433013 & 0.665197 \\
\bottomrule
\end{tabular}
\caption{Collision and von Neumann crossover values used in the diagnostic tests. The entries in the last column are the unique roots of Eq.~\eqref{eq:etavn}. Numerical optimization and stabilizer enumeration are supporting evidence only, not a proof of Conjecture~\ref{conj:eta-renyi}.}
\label{tab:eta-numerics}
\end{table}
The two-use enumeration does not cover nonstabilizer entangled states, and finite-dimensional numerical searches cannot establish tensor-power additivity. We therefore retain Eqs.~\eqref{eq:etaLU}--\eqref{eq:eta-capacity-sandwich} as the rigorous capacity statement above $\eta_\star$.

\section{Relation to prior non-entanglement-breaking and additivity results}

The existence of non-entanglement-breaking Weyl channels is not new, and neither are the collision-norm, covariance, or design-theoretic ingredients taken separately. The new operational separation is collected in Theorem~\ref{thm:balanced-family} and Corollary~\ref{cor:singer-balanced}: along a complete one-parameter path from NPT to entanglement breaking, an infinite sparse anisotropic CPWP family has exactly constant unrestricted capacity, a direct finite-block strong converse, a complete all-use optimizer classification, and an entanglement-assisted advantage whose ratio approaches $2-\varepsilon$.  Theorem~\ref{thm:eta-collision-phase} separately resolves the full collision-entropy bifurcation of the identity--dephasing path.

The boundary between prior ingredients and this combined theorem is important. Shor and Laflamme introduced quantum analogues of classical weight enumerators, and Rains developed related enumerators and duality relations~\cite{ShorLaflamme1997,Rains1998}. The quantities $E_T$ and $C_S$ are subsystem-resolved Weyl-weight bookkeeping of the same general type, and Eq.~\eqref{eq:cumulative} is not claimed as a new enumerator identity. The new use is to couple those cumulative constraints to the eigenvalues of a positive Wigner collision operator, so that signed and entangled inputs are controlled by reduced-state purities and the equality conditions remain visible.

The normalized-projector mechanism for R\'enyi orders $0\leq\alpha\leq2$ and its capacity consequences under covariance were developed by Wolf and Eisert~\cite{Wolf2005}. The present work does not claim that abstract mechanism as new. It identifies the uniform shift supports for which a stabilizer projector saturates the Weyl collision certificate and then resolves the equality conditions at every tensor power.

Fukuda and Holevo~\cite{FukudaHolevo2006} derived the Weyl-covariant output-$2$-norm estimate that evaluates to Eq.~\eqref{eq:mainbound} at one use and showed that saturation by a maximal Weyl subgroup yields multiplicativity. Fukuda and Gour~\cite{FukudaGour2017} later obtained the same numerical estimate as the Weyl-diagonal specialization of a general tensor-stable unital-channel bound. Accordingly, neither the numerical inequality nor its abstract multiplicativity mechanism is the novelty claim. The derivation here rewrites the tensor-power estimate in terms of the reduced-purity identity \eqref{eq:cumulative}. Its new role is to make every inequality and equality condition subsystem-resolved, enabling the difference-set minimax characterization and the all-use rigidity statement.

Recent work on quantum convolutional channels has developed stabilizer-dephasing lower bounds for classical communication and exact formulas for stabilizer-diagonal environments and a nonstabilizer qutrit family~\cite{Xiong2026}. That framework transfers spectra from environment measurements. The present result instead starts from a random-Weyl collision spectrum, selects extremal sparse kernels by cyclic design theory, and proves a fully regularized capacity formula together with uniqueness of all minimum-output R\'enyi states in the proved band and all Holevo-optimal block ensembles, under the strict gap, in every Singer dimension $D=q^2+q+1$ with prime-power $q$.

The operational conclusions established here are not contained in the collision-norm estimates alone. In particular, the one-shot bound treats arbitrary codes directly and yields an explicit strong-converse exponent; the strict-gap and two-axis analyses classify every Holevo-optimal block ensemble, including the degenerate balanced endpoint; and the Fourier-resolved principal-minor criterion locates the exact entanglement-breaking boundary of the balanced path. The entanglement-assisted formula itself follows from the standard covariant-channel mutual-information argument~\cite{Bennett1999,Bennett2002}; its role here is to quantify a scalable separation on the same full interval where basis encoding already attains the unrestricted unassisted capacity.

King's depolarizing-channel capacity theorem~\cite{King2003} is a prominent instance where entangled coding is unnecessary. The current result does not claim novelty merely because its channels preserve entanglement: it adds complete Wigner positivity, sparse anisotropy selected by a combinatorial minimax property, an explicit one-shot strong converse, all-use minimum-output and Holevo-ensemble rigidity, and a simultaneous entanglement-assisted separation.

Amosov~\cite{Amosov2020} proved minimum-output-entropy additivity for Weyl channels satisfying an ordered-deformation condition. His convention $W_{jk}=Z^jX^k$ corresponds, up to an irrelevant phase, to the probability array $\pi_{jk}=q_R(k)h_\eta(j)$ used here. The hypothesis orders that array lexicographically as $\pi_{00}\geq\pi_{10}\geq\cdots\geq\pi_{D-1,0}\geq\pi_{01}\geq\cdots$. Every nonzero shift column begins with $h_{\max}/r$, where $h_{\max}=h_\eta(0)$, and ends with $h_{\min}/r$, where $h_{\min}=h_\eta(j\neq0)$. If the nonzero columns do not form an initial consecutive block, a zero precedes a later entry $h_{\max}/r>0$, contradicting monotonicity. If they do form such a block, then $r>1$ places two nonzero columns consecutively and their boundary requires $h_{\min}\geq h_{\max}$, also impossible for $\eta>0$. Thus the Singer family in the stated nontrivial regime does not fall under Amosov's hypothesis.

\section{Conclusion and open questions}

Our main result is an operational separation, not merely another example of a non-entanglement-breaking channel.  The balanced bi-difference-set path keeps the unrestricted capacity exactly equal to $\log(D/r)$ and retains the same finite-block strong converse from one endpoint to the other, although its Choi state is NPT at every $\varepsilon<1$ and the channel becomes entanglement breaking at $\varepsilon=1$.  The decisive ingredient is not positivity of the phase-space kernel by itself, but a flat nonconstant collision spectrum together with the subsystem purity constraints obeyed by every quantum state.

The equality analysis is equally sharp.  At $\varepsilon=0$, all tensor-power minimizers are products whose local factors may independently use either of two mutually unbiased Weyl bases; for every $\varepsilon>0$, only computational-basis products survive.  The Fourier-resolved Choi minor proves the entanglement boundary, while the exact covariant mutual-information formula shows that $C_{\mathrm E}$ falls from $2C$ to $C$.  Singer parameters realize the construction for every prime power $q$, and at fixed $\varepsilon$ the ratio $C_{\mathrm E}/C$ tends to $2-\varepsilon$.  Thus preservation of channel entanglement, usefulness of nonclassical unassisted encodings, and usefulness of pre-shared entanglement are quantitatively distinct.

For the identity--dephasing path, the all-use collision-entropy optimizer problem is closed on the entire interval, including the two-basis threshold and the Fourier-dominated branch above it.  Proposition~\ref{prop:eta-capacity-sandwich} strengthens the operational picture above $\eta_\star$ by combining two explicit product-basis codes with the collision converse.  The candidate von Neumann transition occurs later, at the unique $\eta_{\rm vN}$, and Conjecture~\ref{conj:eta-renyi} states the precise tensor-power inequality that would close the remaining gap.  We emphasize that the numerical tests do not establish this conjecture: the von Neumann minimum-output entropy and unrestricted classical capacity for $\eta_\star<\eta<1$ remain open.  Extending the construction to even dimensions would also require a different phase-space formalism, and proving comparable rigidity beyond R\'enyi order two will require information not captured by purity alone.

\section*{Author contributions}

S.-W. Ji conceived the project, developed the mathematical arguments, and prepared the original manuscript. Generative AI tools assisted with English editing, LaTeX restructuring, algebraic consistency checks, and the schematic figure. The author reviewed and verified all AI-assisted changes and takes full responsibility for the scientific content.

\section*{Acknowledgements}

This research was supported by the National Research Council of Science \& Technology (NST) grant by the Korea government (MSIT) (CAP22053-200).

\section*{Data and code availability}

No experimental data were created or analyzed. The source package includes the editable source of Fig.~\ref{fig:phase-diagram} and the script \texttt{capacity\_conjecture\_numerics.py}, which reproduces the crossover values and the finite-dimensional diagnostic tests reported around Table~\ref{tab:eta-numerics}.

\bibliographystyle{quantum}
\bibliography{references}

@article{Holevo1998,
  author  = {Holevo, Alexander S.},
  title   = {The capacity of the quantum channel with general signal states},
  journal = {IEEE Transactions on Information Theory},
  volume  = {44},
  number  = {1},
  pages   = {269--273},
  year    = {1998},
  doi     = {10.1109/18.651037}
}

@article{Schumacher1997,
  author  = {Schumacher, Benjamin and Westmoreland, Michael D.},
  title   = {Sending classical information via noisy quantum channels},
  journal = {Physical Review A},
  volume  = {56},
  number  = {1},
  pages   = {131--138},
  year    = {1997},
  doi     = {10.1103/PhysRevA.56.131}
}

@article{Hastings2009,
  author  = {Hastings, Matthew B.},
  title   = {Superadditivity of communication capacity using entangled inputs},
  journal = {Nature Physics},
  volume  = {5},
  pages   = {255--257},
  year    = {2009},
  doi     = {10.1038/nphys1224}
}

@article{Shor2002,
  author  = {Shor, Peter W.},
  title   = {Additivity of the classical capacity of entanglement-breaking quantum channels},
  journal = {Journal of Mathematical Physics},
  volume  = {43},
  number  = {9},
  pages   = {4334--4340},
  year    = {2002},
  doi     = {10.1063/1.1498000}
}

@article{Gross2006,
  author  = {Gross, David},
  title   = {Hudson's theorem for finite-dimensional quantum systems},
  journal = {Journal of Mathematical Physics},
  volume  = {47},
  number  = {12},
  pages   = {122107},
  year    = {2006},
  doi     = {10.1063/1.2393152}
}

@article{Wang2019,
  author  = {Wang, Xin and Wilde, Mark M. and Su, Yuan},
  title   = {Quantifying the magic of quantum channels},
  journal = {New Journal of Physics},
  volume  = {21},
  number  = {10},
  pages   = {103002},
  year    = {2019},
  doi     = {10.1088/1367-2630/ab451d}
}

@article{ShorLaflamme1997,
  author  = {Shor, Peter W. and Laflamme, Raymond},
  title   = {Quantum analog of the {MacWilliams} identities for classical coding theory},
  journal = {Physical Review Letters},
  volume  = {78},
  number  = {8},
  pages   = {1600--1602},
  year    = {1997},
  doi     = {10.1103/PhysRevLett.78.1600}
}

@article{Rains1998,
  author  = {Rains, Eric M.},
  title   = {Quantum weight enumerators},
  journal = {IEEE Transactions on Information Theory},
  volume  = {44},
  number  = {4},
  pages   = {1388--1394},
  year    = {1998},
  doi     = {10.1109/18.681316}
}

@misc{FukudaHolevo2006,
  author        = {Fukuda, Motohisa and Holevo, Alexander S.},
  title         = {On {Weyl}-covariant channels},
  year          = {2006},
  eprint        = {quant-ph/0510148},
  archiveprefix = {arXiv},
  primaryclass  = {quant-ph},
  doi           = {10.48550/arXiv.quant-ph/0510148},
  url           = {https://arxiv.org/abs/quant-ph/0510148}
}

@article{FukudaGour2017,
  author  = {Fukuda, Motohisa and Gour, Gilad},
  title   = {Additive bounds of minimum output entropies for unital channels and an exact qubit formula},
  journal = {IEEE Transactions on Information Theory},
  volume  = {63},
  number  = {3},
  pages   = {1818--1828},
  year    = {2017},
  doi     = {10.1109/TIT.2016.2641455}
}

@article{Wolf2005,
  author  = {Wolf, Michael M. and Eisert, Jens},
  title   = {Classical information capacity of a class of quantum channels},
  journal = {New Journal of Physics},
  volume  = {7},
  pages   = {93},
  year    = {2005},
  doi     = {10.1088/1367-2630/7/1/093}
}

@article{KonigWehner2009,
  author  = {K{\"o}nig, Robert and Wehner, Stephanie},
  title   = {A strong converse for classical channel coding using entangled inputs},
  journal = {Physical Review Letters},
  volume  = {103},
  number  = {7},
  pages   = {070504},
  year    = {2009},
  doi     = {10.1103/PhysRevLett.103.070504}
}

@article{Singer1938,
  author  = {Singer, James},
  title   = {A theorem in finite projective geometry and some applications to number theory},
  journal = {Transactions of the American Mathematical Society},
  volume  = {43},
  number  = {3},
  pages   = {377--385},
  year    = {1938},
  doi     = {10.1090/S0002-9947-1938-1501951-4}
}

@book{Beth1999,
  author    = {Beth, Thomas and Jungnickel, Dieter and Lenz, Hanfried},
  title     = {Design Theory},
  edition   = {2},
  volume    = {1},
  publisher = {Cambridge University Press},
  address   = {Cambridge},
  year      = {1999},
  doi       = {10.1017/CBO9780511549533}
}

@article{King2003,
  author  = {King, Christopher},
  title   = {The capacity of the quantum depolarizing channel},
  journal = {IEEE Transactions on Information Theory},
  volume  = {49},
  number  = {1},
  pages   = {221--229},
  year    = {2003},
  doi     = {10.1109/TIT.2002.806153}
}

@article{Amosov2020,
  author  = {Amosov, Grigori G.},
  title   = {On classical capacity of {Weyl} channels},
  journal = {Quantum Information Processing},
  volume  = {19},
  pages   = {401},
  year    = {2020},
  doi     = {10.1007/s11128-020-02900-5}
}

@article{NathansonRuskai2007,
  author  = {Nathanson, Michael and Ruskai, Mary Beth},
  title   = {Pauli diagonal channels constant on axes},
  journal = {Journal of Physics A: Mathematical and Theoretical},
  volume  = {40},
  number  = {28},
  pages   = {8171--8204},
  year    = {2007},
  doi     = {10.1088/1751-8113/40/28/S22}
}

@article{Siudzinska2020,
  author  = {Siudzi{\'n}ska, Katarzyna},
  title   = {Classical capacity of the generalized {Pauli} channels},
  journal = {Journal of Physics A: Mathematical and Theoretical},
  volume  = {53},
  number  = {44},
  pages   = {445301},
  year    = {2020},
  doi     = {10.1088/1751-8121/abb276}
}

@article{Rehman2018,
  author  = {ur Rehman, Junaid and Jeong, Youngmin and Kim, Jeong San and Shin, Hyundong},
  title   = {Holevo capacity of discrete {Weyl} channels},
  journal = {Scientific Reports},
  volume  = {8},
  pages   = {17457},
  year    = {2018},
  doi     = {10.1038/s41598-018-35777-7}
}

@article{Peres1996,
  author  = {Peres, Asher},
  title   = {Separability criterion for density matrices},
  journal = {Physical Review Letters},
  volume  = {77},
  number  = {8},
  pages   = {1413--1415},
  year    = {1996},
  doi     = {10.1103/PhysRevLett.77.1413}
}

@article{Horodecki1996,
  author  = {Horodecki, Micha{\l} and Horodecki, Pawe{\l} and Horodecki, Ryszard},
  title   = {Separability of mixed states: necessary and sufficient conditions},
  journal = {Physics Letters A},
  volume  = {223},
  number  = {1--2},
  pages   = {1--8},
  year    = {1996},
  doi     = {10.1016/S0375-9601(96)00706-2}
}

@article{Bennett1999,
  author  = {Bennett, Charles H. and Shor, Peter W. and Smolin, John A. and Thapliyal, Ashish V.},
  title   = {Entanglement-assisted classical capacity of noisy quantum channels},
  journal = {Physical Review Letters},
  volume  = {83},
  number  = {15},
  pages   = {3081--3084},
  year    = {1999},
  doi     = {10.1103/PhysRevLett.83.3081}
}

@article{Bennett2002,
  author  = {Bennett, Charles H. and Shor, Peter W. and Smolin, John A. and Thapliyal, Ashish V.},
  title   = {Entanglement-assisted capacity of a quantum channel and the reverse {Shannon} theorem},
  journal = {IEEE Transactions on Information Theory},
  volume  = {48},
  number  = {10},
  pages   = {2637--2655},
  year    = {2002},
  doi     = {10.1109/TIT.2002.802612}
}

@article{HorodeckiShorRuskai2003,
  author  = {Horodecki, Micha{\l} and Shor, Peter W. and Ruskai, Mary Beth},
  title   = {General entanglement breaking channels},
  journal = {Reviews in Mathematical Physics},
  volume  = {15},
  number  = {6},
  pages   = {629--641},
  year    = {2003},
  doi     = {10.1142/S0129055X03001709}
}

@misc{Xiong2026,
  author        = {Xiong, Chunhe and Kim, Sunho and Zhang, Qing-hua and Fei, Shao-ming},
  title         = {Mean-state entropy hierarchies and classical communication through quantum convolutions},
  year          = {2026},
  eprint        = {2607.16653},
  archiveprefix = {arXiv},
  primaryclass  = {quant-ph},
  doi           = {10.48550/arXiv.2607.16653},
  url           = {https://arxiv.org/abs/2607.16653}
}

\end{document}